\documentclass[11pt,a4paper, oneside]{article}

\usepackage{amsmath}
\usepackage{amsfonts}
\usepackage{amssymb}
\usepackage{amscd}
\usepackage{latexsym}
\usepackage{mathrsfs}
\usepackage{amsthm}

\newcommand{\CC}{\mathbb{C}}
\newcommand{\RR}{\mathbb{R}}
\newcommand{\ZZ}{\mathbb{Z}}

\def\S{\mathbb{S}}
\newcommand{\HH}{\mathcal{H}}
\newcommand{\dom}{\mathop{\mathcal{D}}}

\newcommand{\im}{\mathop{\mathrm{im}}}
\newcommand{\tr}{\mathop{\mathrm{tr}}}
\newcommand{\vol}{\mathop{\mathrm{vol}}}
\def\Tr{\mathrm{Tr}}

\def\Y{Y}
\def\X{X}

\def\E{\mathcal{E}}
\def\J{\mathcal{J}}
\def\Q{\mathcal{Q}}
\def\A{\mathcal{A}}
\def\B{\mathcal{B}}

\def\ind{\mathrm{ind}}
\newcommand{\bR}{\overline{\R}}
\newcommand{\bRp}{\overline{\R_+}}
\def\K{\mathcal K}

\newcommand{\R}{{\mathbb R}}
\newcommand{\T}{{\mathbb T}}
\newcommand{\Z}{{\mathbb Z}}
\newcommand{\N}{{\mathbb N}}
\newcommand{\C}{{\mathbb C}}

\def\f{\mathfrak{f}}

\def\Var{\mathrm{Var}}
\def\sign{\;\!\mathrm{sign}}
\def\ud{{\textstyle \frac{1}{2}}}

\def\i{{\tt i }}
\def\q{{\tt q}}
\def\b{{\tt b}}
\def\d{{\tt d}}
\def\j{{\tt j}}
\def\wind{\mathrm{wind}}
\def\ch{\mathrm{ch}}

\def\dd{\mathrm{d}}
\def\det{\mathrm{det}}
\def\AB{A\!B}
\def\CD{C\!D}
\def\F21{{}_2F_1}
\def\2{\mathfrak{int}}

\def\P{{\mathrm P}}

\def\Arg{\theta}

\newcommand{\bOmega}{{\bf \Omega}}
\newcommand{\bGamma}{{\bf \Gamma}}
\newcommand{\bP}{{\bf P}}
\newtheorem{theorem}{Theorem}
\newtheorem{prop}[theorem]{Proposition}
\newtheorem{lemma}[theorem]{Lemma}
\newtheorem{corol}[theorem]{Corollary}
\newtheorem{defin}[theorem]{\bf Definition}
\theoremstyle{definition}

\newtheorem{rem}[theorem]{\bf Remark}
\newtheorem{example}[theorem]{\bf Example}

\begin{document}

\title{\bf Levinson's theorem and higher degree traces \\
for Aharonov-Bohm operators}

\author{Johannes Kellendonk$^1$, Konstantin Pankrashkin$^2$ and Serge Richard$^3$}
\date{\small}
\maketitle

\begin{quote}
\begin{itemize}
\item[$^1$] Universit\'e de Lyon, Universit\'e Lyon I, CNRS UMR5208, Institut Camille Jordan, 43 blvd du 11 novembre 1918, 69622 Villeurbanne Cedex, France; \\
E-mail: {\tt kellendonk@math.univ-lyon1.fr}
\item[$^2$] Laboratoire de Math\'ematiques d'Orsay, CNRS UMR 8628, Universit\'e Paris-Sud XI, B\^atiment 425, 91405 Orsay Cedex, France;\\
E-mail: {\tt konstantin.pankrashkin@math.u-psud.fr}
\item[$^3$] Graduate School of Pure and Applied Sciences,
University of Tsukuba, 1-1-1 Tennodai, Tsukuba,
Ibaraki 305-8571, Japan; \\
E-mail: {\tt richard@math.univ-lyon1.fr}
\\[\smallskipamount]
On leave from Universit\'e de Lyon, Universit\'e Lyon I, CNRS UMR5208, Institut Camille Jordan, 43 blvd du 11 novembre 1918, 69622 Villeurbanne Cedex, France
\end{itemize}
\end{quote}

\begin{abstract}
We study Levinson type theorems for the family of Aharonov-Bohm models from different perspectives.
The first one is purely analytical involving the explicit calculation of the wave-operators and allowing
to determine precisely the various contributions to the left hand side of Levinson's theorem, namely those
due to the scattering operator, the terms at $0$-energy and at energy $+\infty$. The second one is based on non-commutative topology revealing the topological nature of Levinson's theorem. We then include the parameters of the family into the topological description obtaining a new type of Levinson's theorem, a higher degree Levinson's theorem. In this context, the Chern number of a bundle defined by a family of projections on bound states is explicitly computed and related to the result of a $3$-trace applied on the scattering part of the model.
\end{abstract}


\section{Introduction}

In recent work \cite{KR1,KRx,KR1/2,KR3D,RT} it was advocated that
Levinson's theorem is of topological nature, namely that it should be viewed as an index theorem.
The relevant index theorem occurs naturally in the framework of non-commutative  topology, that is,
$C^*$-algebras, their $K$-theory and higher traces (unbounded cyclic cocycles). The analytical hypothesis
which has to be fulfilled for the index theoretic formulation to hold is that
 the wave operators of the scattering system lie in a certain $C^*$-algebra.  In the examples considered until now, the index theorem substantially extends the usual Levinson's theorem which relates the number of bound states of a physical system to an expression depending on the scattering part of the system. In particular it sheds new light on the corrections due to resonances and on the regularization which are often involved in the proof of this relation. It also emphasizes the influence of the restriction of the waves operators at thresholds energies.

In the present paper we extend these investigations in two directions. On the one hand, we apply the general idea for the first time to a magnetic system.
Indeed, the Aharonov-Bohm operators describe a two-dimensional physical system involving a singular magnetic field located at the origin and perpendicular to the plane of motion. On the other hand, due to the large number of parameters present in this model, we can develop a new topological equality involving higher degree traces. Such an equality, which we call a {\em higher degree Levinson's theorem}, extends naturally the usual Levinson's theorem (which corresponds to a relation between an $0$-trace and a $1$-trace) and it is apparently the first time that a relation between a $2$-trace and a $3$-trace is put into evidence in a physical context. While the precise physical meaning of this equality deserves more investigations, we have no doubt that it can play a role in the theory of topological transport and/or of adiabatic pumping \cite{Graf}.

Let us describe more precisely the content of this paper. In Section \ref{recall} we recall the contruction of the Aharonov-Bohm operators and present part of the results obtained in \cite{PR}. Earlier references for the basic properties of these operators are \cite{AT,AB,DS,R,Rui}. In particular, we recall the explicit expressions for the wave operators in terms of functions of the free Laplacian and of the generator of the dilation group in $\R^2$. Let us mention that the theory of boundary triples, as presented in \cite{BGP} was extensively used in reference \cite{PR} for the computation of these explicit expressions.

In Section \ref{secLev} we state and prove a version of Levinson's theorem adapted to our model, see Theorem \ref{Lev0}. It will become clear at that moment that a naive approach of this theorem involving only the scattering operator would lead to a completely wrong result. Indeed, the corrections due to the restriction of the wave operators at $0$-energy and at energy equal to $+\infty$ will be explicitly computed. Adding these different contributions leads to a first proof of Levinson's theorem. All the various situations, which depend on the parameters related to the flux of the magnetic field and to the description of the self-adjoint extensions, are summarized in Section \ref{Sectionfinal}. Let us stress that this proof is rather lengthy but that it leads to a very precise result. Note that up to this point, no $C^*$-algebraic knowledge is required, all proofs are purely analytical.

The last two sections of the paper contain the necessary algebraic framework, the two topological statements and their proofs. So Section \ref{secK} contains a very short introduction to $K$-theory, cyclic cohomology, $n$-traces, Connes' pairing and the dual boundary maps. Obviously, only the very few necessary information on these subjects is presented, and part of the constructions are over-simplified. However, the authors tried to give a flavor of this necessary background for non-experts, but any reader familiar with these constructions can skip Section \ref{secK} without any lost of understanding in the last part of the paper.

In the first part of Section \ref{secAlgebra}, we construct a suitable $C^*$-algebra $\E$ which contains the wave operators. For computational reasons, this algebra should neither be too small nor too large. In the former case, the computation of its quotient by the ideal of compact operators would be too difficult and possibly not understandable, in the latter case the deducible information would become too vague. In fact, the algebra we propose is very natural once the explicit form of the wave operators is known. Once the quotient of the algebra $\E$ by the compact operators is computed, the new topological version of Levinson's can be stated. This is done in Theorem \ref{Ktheo} and in that case its proof is contained in a few lines. Note furthermore that there is a big difference between Theorem \ref{Lev0} and the topological statement (and its corollary). In the former case, the proof consisted in checking that the sum of various explicit contributions is equal to the number of bound states of the corresponding system. In the latter case, the proof involves a topological argument and it clearly shows the topological nature of Levinson's theorem. However, the statement is global, and the contributions due to the scattering operator and to the restrictions at $0$-energy and at energy $+\infty$ can not be distinguished. For that reason, both approach are complementary. Note that the topological approach opens the way towards generalisations which could hardly be guessed from the purely analytical approach.

Up to this point, the flux of the magnetic field as well as the parameters involved in the description of the self-adjoint extension were fixed. In the second topological statement, we shall consider a smooth boundaryless submanifold of the parameters space and perform some computations as these parameters vary on the manifold. More precisely, we first state an equality between a continuous family of projections on the bound states and the image through the index map of a continuous family of unitary operators deduced from the wave operators, see Theorem \ref{thm-ENN}. These unitary operators contain a continuous family of scattering operators, but also the corresponding continuous family of restrictions at energies $0$ and $+\infty$. Note that this result is still abstract, in the sense that it gives an equality between an equivalent class in the $K_0$-theory related to the bounded part of the system with an equivalent class in the $K_1$-theory related to the scattering part
  of the system, but nothing prevents this equality from being trivial in the sense that it yields $0=0$.

In the final part of the paper, we choose a $2$-dimensional submanifold and show that the second topological result is not trivial. More precisely, we explicitly compute the pairings of the $K$-equivalent classes with their respective higher degree traces. On the one hand this leads to the computation of the Chern number of a bundle defined by the family of projections. For the chosen manifold this number is equal to $1$, and thus is not trivial. By duality of the boundary maps, it follows that the natural $3$-trace applied on the family of unitary operators is also not trivial. The resulting statement is provided in Proposition \ref{propfinal}. Note that this statement is again global. A distinction of each contribution could certainly be interesting for certain applications, but its computation could be rather tedious and therefore no further investigations have been performed in that direction.

\section*{Acknowledgements}
S. Richard was supported by the Swiss National Science Foundation and is now supported by the Japan Society for the Promotion of Sciences.

\section{The Aharonov-Bohm model}\label{recall}

In this section, we briefly recall the construction of the Aharonov-Bohm operators and
present a part of the results obtained in \cite{PR} to which we refer for details.
We also mention \cite{AT,DS,Rui} for earlier works on these operators.

\subsection{The self-adjoint extensions}\label{ssec21}

Let $\HH$ denote the Hilbert space $L^2(\RR^2)$ with its scalar product
$\langle \cdot,\cdot\rangle$ and its norm $\|\cdot \|$.
For any $\alpha \in (0,1)$, we set $A_\alpha: \RR^2\setminus\{0\} \to \RR^2$ by
\begin{equation*}
A_\alpha(x,y)= -\alpha \Big(\frac{-y}{x^2+y^2},
\frac{x}{x^2+y^2}
\Big),
\end{equation*}
corresponding formally to the magnetic field $B=\alpha\delta$ ($\delta$ is the Dirac
delta function), and consider the operator
\begin{equation*}
H_\alpha:=(-i\nabla -A_\alpha)^2,
\qquad \dom(H_\alpha)=C_c^\infty\big(\RR^2\setminus\{0\}\big)\ .
\end{equation*}
Here $C_c^\infty(\Xi)$ denotes the set of smooth functions on $\Xi$ with compact support.
The closure of this operator in $\HH$, which is denoted by the same symbol,
is symmetric and has deficiency indices $(2,2)$.

We briefly recall the parametrization of the self-adjoint extensions of $H_\alpha$ from \cite{PR}.
Some elements of the domain of the adjoint operator $H_\alpha^*$ admit singularities
at the origin. For dealing with them, one defines linear functionals $\Phi_0$, $\Phi_{-1}$, $\Psi_0$, $\Psi_{-1}$
on $\dom(H_\alpha^*)$ such that for $\f\in\dom(H_\alpha^*)$ one has, with $\theta \in [0,2\pi)$ and $r \to 0_+$,
\[
2\pi \f(r\cos\theta,r\sin\theta)= \Phi_0(\f)r^{-\alpha}+\Psi_0(\f) r^\alpha
+e^{-i\theta} \Big(
\Phi_{-1}(\f)r^{\alpha-1}+\Psi_{-1}(\f) r^{1-\alpha}
\Big) +O(r).
\]
The family of all self-adjoint extensions of the operator $H_\alpha$ is then indexed by two matrices
$C,D \in M_2(\CC)$ which satisfy the following conditions:
\begin{equation}
\label{eq-mcd}
\text{(i) $CD^*$ is self-adjoint,\qquad  (ii) $\det(CC^* + DD^*)\neq 0$,}
\end{equation}
and the corresponding extensions $H^{\CD}_\alpha$ are the restrictions of $H_\alpha^*$
onto the functions $\f$ satisfying the boundary conditions
\[
C \begin{pmatrix}
\Phi_0(\f)\\ \Phi_{-1}(\f)
\end{pmatrix}
=2D \begin{pmatrix}
\alpha \Psi_0(\f)\\  (1-\alpha)\Psi_{-1}(\f)
\end{pmatrix}.
\]
For simplicity, we call \emph{admissible} a pair of
matrices $(C,D)$  satisfying the above conditions.

\begin{rem}\label{1to1}
The parametrization of the self-adjoint extensions of $H_\alpha$ with all admissible
pairs $(C,D)$ is very convenient but highly none unique.
At a certain point, it will be useful to have a one-to-one parametrization of all
self-adjoint extensions.
So, let us consider $U \in U(2)$ and set
\begin{equation*}
C(U) := {\textstyle \frac{1}{2}}(1-U) \quad \hbox{ and }
\quad D(U) = {\textstyle \frac{i}{2}}(1+U).
\end{equation*}
It is easy to check that $C(U)$ and $D(U)$ satisfy both conditions \eqref{eq-mcd}.
In addition, two different elements $U,U'$ of $U(2)$ lead to two different self-adjoint
operators $H_\alpha^{C(U)\;\!D(U)}$ and $H_\alpha^{C(U')\;\!D(U')}$, {\it cf.}~\cite{Ha}.
Thus, without ambiguity we can write $H_\alpha^U$ for the operator $H_\alpha^{C(U)\;\!D(U)}$.
Moreover, the set $\{H_\alpha^U\mid U \in U(2)\}$ describes all
self-adjoint extensions of $H_\alpha$.
Let us also mention that the normalization of the above maps has been chosen such that
$H_\alpha^{-1}\equiv H_\alpha^{10}= H_\alpha^{\AB}$ which corresponds to the standard Aharonov-Bohm operator
studied in \cite{AB,Rui}.
\end{rem}

The essential spectrum of $H^{\CD}_\alpha$ is absolutely continuous and covers the positive half line $[0,+\infty)$.
The discrete spectrum consists of at most two negative eigenvalues. More precisely, the number of negative eigenvalues
of $H^{\CD}_\alpha$ coincides with the number of negative eigenvalues of the matrix $CD^*$.

The negative eigenvalues are the real negative solutions of the equation
\[
\det\big(
DM(z)-C\big)=0
\]
where $M(z)$ is, for $z<0$,
\[
M(z)=- \frac{2}{\pi} \sin (\pi \alpha)\,
\begin{pmatrix}
\Gamma(1-\alpha)^2 \Big( -\dfrac{z}{4}\Big)^\alpha & 0 \\
0& \Gamma(\alpha)^2 \Big( -\dfrac{z}{4}\Big)^{1-\alpha}
\end{pmatrix},
\]
and there exists an injective map $\gamma(z):\CC^2\to\HH$ depending continuously on $z \in \C \setminus [0,+\infty)$ and calculated explicitly in \cite{PR}
such that for each $z<0$ one has $\ker(H^{\CD}_\alpha-z)=\gamma(z)\ker\big(
DM(z)-C\big)$.

\subsection{Wave and scattering operators}

One of the main result of \cite{PR} is an explicit description of the wave operators.
We shall recall this result below, but we first need
to introduce the decomposition of the Hilbert space $\HH$ with respect to the
spherical harmonics.
For any $m \in \ZZ$, let $\phi_m$ be the complex function defined by
$[0,2\pi)\ni \theta \mapsto \phi_m(\theta):= \frac{e^{im\theta}}{\sqrt{2\pi}}$.
One has then the canonical isomorphism
\begin{equation}\label{decomposition}
\HH \cong \bigoplus_{m \in \ZZ} \HH_r \otimes [\phi_m] \ ,
\end{equation}
where $\HH_r:=L^2(\RR_+, r\;\!\dd r)$ and $[\phi_m]$ denotes the one dimensional space spanned by $\phi_m$.
For shortness, we write $\HH_m$ for $\HH_r \otimes [\phi_m]$, and often consider it as a subspace of $\HH$.
Let us still set $\HH_\2:=\HH_0\oplus\HH_{-1}$ which is clearly isomorphic to $\HH_r\otimes \C^2$.

Let us also recall that the unitary dilation group
$\{U_\tau\}_{\tau \in \RR}$ is defined on any $\f \in \HH$ and $x \in \RR^2$ by
\begin{equation*}
[U_\tau \f](x) = e^\tau \f(e^\tau x)\ .
\end{equation*}
Its self-adjoint generator $A$ is formally given by
$\frac{1}{2}(X\cdot (-i\nabla) + (-i\nabla)\cdot X)$, where $X$ is the position operator and $-i\nabla$
is its conjugate operator. All these operators are essentially self-adjoint on the Schwartz space on $\RR^2$.
Clearly, the group of dilations as well as its generator leave each subspace $\HH_m$ invariant.

Let us now consider the wave operators
\begin{equation*}
\Omega^{\CD}_-:=\Omega_-(H^{\CD}_\alpha,H_0)=s-\lim_{t\to - \infty}e^{itH^{\CD}_\alpha }\;\!e^{-itH_0 }\ .
\end{equation*}
where $H_0:=-\Delta$.
It is well known that for any admissible pair $(C,D)$ the operator $\Omega_\pm^{\CD}$ is reduced by
the decomposition $\HH=\HH_\2 \oplus \HH_\2^\bot$ and that
$\Omega_-^{\CD}|_{\HH_\2^\bot} = \Omega_-^{\AB}|_{\HH_\2^\bot}$.
The restriction to $\HH_\2^\bot$ is further reduced by the decomposition \eqref{decomposition}
and it is proved in \cite[Prop.~10]{PR} that the channel wave operators satisfy for each $m \in \Z$,
\begin{equation*}
\Omega_{-,m}^{\AB} = \varphi_m^-(A)\ ,
\end{equation*}
with $\varphi_m^-$ explicitly given for $x \in \R$ by
\begin{equation*}
\varphi^-_m(x):=e^{i\delta_m^\alpha}\;\!
\frac{\Gamma\big(\frac{1}{2}(|m|+1+ix)\big)}{\Gamma\big(\frac{1}{2}(|m|+1-ix)\big)}
\;\!
\frac{\Gamma\big(\frac{1}{2}(|m+\alpha|+1-ix)\big)}{\Gamma\big(\frac{1}{2}(|m+\alpha|+1+ix)\big)}
\end{equation*}
and
\begin{equation*}
\delta_m^\alpha = \hbox{$\frac{1}{2}$}\pi\big(|m|-|m+\alpha|\big)
=\left\{\begin{array}{rl}
-\hbox{$\frac{1}{2}$}\pi\alpha & \hbox{if }\ m\geq 0 \\
\hbox{$\frac{1}{2}$}\pi\alpha & \hbox{if }\ m< 0
\end{array}\right.\ .
\end{equation*}
It is also proved in \cite[Thm.~11]{PR} that
\begin{equation}\label{yoyo}
\Omega_-^{\CD}|_{\HH_\2} = \Big(
\begin{smallmatrix}\varphi^-_0(A) & 0 \\ 0 & \varphi^-_{-1}(A) \end{smallmatrix}\Big) +
\Big(
\begin{smallmatrix}\tilde{\varphi}_0(A) & 0 \\ 0 & \tilde{\varphi}_{-1}(A) \end{smallmatrix}\Big)
\widetilde{S}^{\CD}_\alpha\big(\sqrt{H_0}\big)
\end{equation}
with $\tilde{\varphi}_m(x)$ given for $m \in \{0,-1\}$ by
\begin{eqnarray*}
\frac{1}{2\pi}\;\!e^{-i\pi|m|/2} \;\!e^{\pi x/2}\;\!
\frac{\Gamma\big(\frac{1}{2}(|m|+1+ix)\big)}{\Gamma
\big(\frac{1}{2}(|m|+1-ix)\big)}  \Gamma\big(\ud(1+|m+\alpha|-ix)\big)
\;\!\Gamma\big(\ud(1-|m+\alpha|-ix)\big)\ .
\end{eqnarray*}

Clearly, the functions $\varphi^-_m$ and $\tilde{\varphi}_m$ are continuous on $\R$. Furthermore,
these functions admit limits at $\pm \infty$: $\varphi^-_m(-\infty)=1$, $\varphi^-_m(+\infty)=e^{2i\delta^\alpha_m}$,
$\tilde{\varphi}_m(-\infty)=0$ and $\tilde{\varphi}_m(+\infty)=1$.
Note also that the expression for the function $\widetilde{S}^{\CD}_\alpha(\cdot)$ is given
for $\kappa \in \R_+$ by
\begin{eqnarray*}
\widetilde{S}_\alpha^{\CD}(\kappa)
&:=& 2i\sin(\pi\alpha)
\left(\begin{matrix}
\frac{\Gamma(1-\alpha)\;\!e^{-i\pi\alpha/2}}{2^\alpha}\;\!\kappa^{\alpha} & 0 \\
0 & \frac{ \Gamma(\alpha)\;\!e^{-i\pi(1-\alpha)/2}}{2^{1-\alpha}}\;\!\kappa^{(1-\alpha)}
\end{matrix}\right)  \\
&&\cdot \left( D\,
\left(\begin{matrix}
\frac{\Gamma(1-\alpha)^2 \;\!e^{ -i\pi\alpha}}{4^\alpha}\;\!\kappa^{2\alpha} & 0 \\
0& \frac{\Gamma(\alpha)^2\;\! e^{ -i\pi(1-\alpha)}}{4^{1-\alpha}}\;\!\kappa^{2(1-\alpha)}
\end{matrix}\right)
+\frac{\pi}{2\sin(\pi\alpha)}C\right)^{-1} D \\
&& \cdot
\left(\begin{matrix}
\frac{ \Gamma(1-\alpha)\;\!e^{-i\pi\alpha/2}}{2^\alpha}\;\!\kappa^{\alpha} & 0 \\
0 & -\frac{ \Gamma(\alpha)\;\!e^{-i\pi(1-\alpha)/2}}{2^{1-\alpha}}\;\!\kappa^{(1-\alpha)}
\end{matrix}\right)\ .
\end{eqnarray*}

As usual, the scattering operator is defined by the formula
\begin{equation*}
S^{\CD}_\alpha:=\big[\Omega^{\CD}_+\big]^* \Omega^{\CD}_-.
\end{equation*}
Then, the relation between this operator and $\widetilde{S}^{\CD}_\alpha$ is of the form
\begin{equation} \label{eq-stilde}
S_\alpha^{\CD}|_{\HH_\2}=
S_\alpha^{\CD}(\sqrt{H_0})
\quad\hbox{with}\quad
S_\alpha^{\CD}(\kappa):=\begin{pmatrix}
e^{-i\pi\alpha} & 0 \\
0 & e^{i\pi\alpha}
\end{pmatrix}
+ \widetilde{S}_\alpha^{\CD}(\kappa)\ .
\end{equation}
The following result has been obtained in \cite[Prop.~13]{PR} and will be necessary further on:

\begin{prop}\label{propSurS}
The map
\begin{equation*}
\RR_+\ni \kappa \mapsto S^{\CD}_\alpha(\kappa) \in U(2)
\end{equation*}
is continuous and has explicit asymptotic values for $\kappa=0$ and $\kappa = +\infty$.
More explicitly, depending on $C, D$ and $\alpha$ one has:
\begin{enumerate}
\item[i)] If $D=0$, then $S^{\CD}_\alpha(\kappa)=\left(\begin{smallmatrix}
e^{-i\pi \alpha} & 0\\
0 & e^{i\pi\alpha}
\end{smallmatrix}\right)$,

\item[ii)] If $\det(D)\neq 0$, then
$S^{\CD}_\alpha(+\infty)=\left(\begin{smallmatrix}
e^{i\pi \alpha} & 0\\
0 & e^{-i\pi\alpha}
\end{smallmatrix}\right)$,

\item[iii)] If
$\dim[\ker(D)]=1$ and $\alpha =1/2$, then
$S^{\CD}_\alpha(+\infty)=(2\P -1)\;\!\left(\begin{smallmatrix}
i & 0\\
0 & -i
\end{smallmatrix}\right)$,
where $\P$ is the orthogonal projection onto $\ker(D)^\bot$,

\item[iv)] If $\ker(D)= \left(\begin{smallmatrix}
\CC\\ 0 \end{smallmatrix}\right)$ or if
$\dim[\ker(D)]=1$, $\alpha < 1/2$ and $\ker(D)\neq \left(\begin{smallmatrix}
0\\ \CC
\end{smallmatrix}\right)$,
then
$S^{\CD}_\alpha(+\infty)=\left(\begin{smallmatrix}
e^{-i\pi \alpha} & 0\\
0 & e^{-i\pi\alpha}
\end{smallmatrix}\right)$,

\item[v)] If $\ker(D)= \left(\begin{smallmatrix}
0\\ \CC
\end{smallmatrix}\right)$ or if
$\dim[\ker(D)]=1$, $\alpha > 1/2$ and $\ker(D)\neq \left(\begin{smallmatrix}
\CC\\ 0
\end{smallmatrix}\right)$,
then
$S^{\CD}_\alpha(+\infty)=\left(\begin{smallmatrix}
e^{i\pi \alpha} & 0\\
0 & e^{i\pi\alpha}
\end{smallmatrix}\right)$.
\end{enumerate}

Furthermore,
\begin{enumerate}

\item[a)] If  $C=0$, then
$S^{\CD}_\alpha(0)=\left(\begin{smallmatrix}
e^{i\pi \alpha} & 0\\
0 & e^{-i\pi\alpha}
\end{smallmatrix}\right)$,

\item[b)] If  $\det(C)\ne 0$, then
$S^{\CD}_\alpha(0)=\left(\begin{smallmatrix}
e^{-i\pi \alpha} & 0\\
0 & e^{i\pi\alpha}
\end{smallmatrix}\right)$,

\item[c)] If $\dim[\ker (C)]=1$ and $\alpha=1/2$, then
$S^{\CD}_\alpha(0)=
(1-2\Pi)\left(\begin{smallmatrix}
i & 0\\
0 & -i
\end{smallmatrix}\right)$,
where $\Pi$ is the orthogonal projection on $\ker(C)^\perp$.

\item[d)] If $\ker (C)=\left(\begin{smallmatrix}0\\ \CC \end{smallmatrix}\right)$ or if  $\dim[\ker(C)]=1$,
$\alpha>1/2$ and
$\ker (C)\ne \left(\begin{smallmatrix}\CC \\ 0 \end{smallmatrix}\right)$,
then
$S^{\CD}_\alpha(0)=\left(\begin{smallmatrix}
e^{-i\pi \alpha} & 0\\
0 & e^{-i\pi\alpha}
\end{smallmatrix}\right)$,

\item[e)] If $\ker (C)=\left(\begin{smallmatrix} \CC \\ 0\end{smallmatrix}\right)$ or if
$\dim[\ker (C)]=1$, $\alpha<1/2$ and $\ker (C)\ne \left(\begin{smallmatrix}0 \\ \CC \end{smallmatrix}\right)$,
then
$S^{\CD}_\alpha(0)=\left(\begin{smallmatrix}
e^{i\pi \alpha} & 0\\
0 & e^{i\pi\alpha}
\end{smallmatrix}\right)$.
\end{enumerate}
\end{prop}

\section{The $\boldsymbol{0}$-degree Levinson's theorem, a pedestrian approach}\label{secLev}

In this section, we state a Levinson's type theorem adapted to our model. The proof is quite ad-hoc and will look like a recipe, but a much more conceptual one will be given subsequently.
The main interest in this pedestrian approach is that it shows the importance of the restriction of the wave operators at $0$-energy and at energy equal to $+\infty$.
Let us remind the reader interested in the algebraic approach that the present proof can be skipped without any lost of understanding in the following sections.

Let us start by considering again the expression \eqref{yoyo} for the operator $\Omega_-^{\CD}|_{\HH_\2}$. It follows from the explicit expressions for the functions $\varphi^-_m$, $\tilde{\varphi}_m$ and $\widetilde{S}^{\CD}_\alpha$ that $\Omega_-^{\CD}|_{\HH_\2}$ is a linear combination of product of functions of two non-commutating operators with functions that are respectively continuous on $[-\infty,\infty]$ and on $[0,\infty]$ and which take values in $M_2(\C)$. For a reason that will become limpid in the algebraic framework, we shall consider the restrictions of these products of functions on the asymptotic values of the closed intervals. Namely, let us set for $x \in \R$ and $\kappa\in \R_+$
\begin{eqnarray}
\label{Gma1}\Gamma_1(C,D,\alpha,x)&:=&
\Big(
\begin{smallmatrix}\varphi^-_0(x) & 0 \\ 0 & \varphi^-_{-1}(x) \end{smallmatrix}\Big) +
\Big(
\begin{smallmatrix}\tilde{\varphi}_0(x) & 0 \\ 0 & \tilde{\varphi}_{-1}(x) \end{smallmatrix}\Big)
\widetilde{S}^{\CD}_\alpha(0)\ ,\\
\label{Gma2}\Gamma_2(C,D,\alpha,\kappa)&:=&S^{\CD}_\alpha(\kappa)\ ,\\
\label{Gma3}\Gamma_3(C,D,\alpha,x)&:=&
\Big(
\begin{smallmatrix}\varphi^-_0(x) & 0 \\ 0 & \varphi^-_{-1}(x) \end{smallmatrix}\Big) +
\Big(
\begin{smallmatrix}\tilde{\varphi}_0(x) & 0 \\ 0 & \tilde{\varphi}_{-1}(x) \end{smallmatrix}\Big)
\widetilde{S}^{\CD}_\alpha(+\infty)\ ,\\
\label{Gma4}\Gamma_4(C,D,\alpha,\kappa)&:=& 1.
\end{eqnarray}
We also set
\begin{equation}\label{fctGamma}
\Gamma(C,D,\alpha,\cdot):= \big(\Gamma_1(C,D,\alpha,\cdot), \Gamma_2(C,D,\alpha,\cdot),\Gamma_3(C,D,\alpha,\cdot),
\Gamma_4(C,D,\alpha,\cdot)\big)\ .
\end{equation}

In fact, $\Gamma(C,D,\alpha,\cdot)$ is a continuous function on the edges $\square$ of the square $[0,\infty]\times[-\infty,\infty]$ and takes values in $U(2)$. Thus, since $\Gamma(C,D,\alpha,\cdot) \in C\big(\square,U(2)\big)$,
we can define the winding number $\wind\big[\Gamma(C,D,\alpha,\cdot)\big]$ of the map
\begin{equation*}
\square \ni \zeta \mapsto \det [\Gamma(C,D,\alpha,\zeta)]\in \T
\end{equation*}
with orientation of $\square$ chosen clockwise. Here $\T$ denotes the set of complex numbers of modulus $1$.
The following statement is our Levinson's type theorem.

\begin{theorem}\label{Lev0}
For any $\alpha\in (0,1)$ and any admissible pair $(C,D)$ one has
\begin{equation*}
\wind \big[\Gamma(C,D,\alpha,\cdot)\big] = - \# \sigma_p(H_\alpha^{\CD}) =
-\#\{\hbox{negative eigenvalues of }CD^*\}\ .
\end{equation*}
\end{theorem}

\begin{proof}
The first equality is proved below by a case-by-case study.
The equality between the cardinality of $\sigma_p(H_\alpha^{\CD})$ and the number of
negative eigenvalues of the matrix $CD^*$ has been shown in \cite[Lem.~4]{PR}.
\end{proof}

We shall now calculate separately the contribution to the winding number from the functions $\Gamma_1(C,D,\alpha,\cdot)$, $\Gamma_2(C,D,\alpha,\cdot)$ and $\Gamma_3(C,D,\alpha,\cdot)$. The contribution due to the scattering operator is the one given by $\Gamma_2(C,D,\alpha,\cdot)$. It will be rather clear that a naive approach of Levinson's theorem involving only the contribution of the scattering operator would lead to a completely wrong result. The final results are presented in Section \ref{Sectionfinal}.

\subsection{Contributions of $\Gamma_1(C,D,\alpha,\cdot)$ and $\Gamma_3(C,D,\alpha,\cdot)$}

In this section we calculate the contributions due to $\Gamma_1(C,D,\alpha,\cdot)$ and $\Gamma_3(C,D,\alpha,\cdot)$ which were introduced in \eqref{Gma1} and \eqref{Gma3}. For that purpose, recall first the relation
\begin{equation*}
S_\alpha^{\CD}(\kappa):=\left(\begin{smallmatrix}
e^{-i\pi\alpha} & 0 \\
0 & e^{i\pi\alpha}
\end{smallmatrix}\right)
+ \widetilde{S}_\alpha^{\CD}(\kappa)\ .
\end{equation*}
Since $S^{\CD}_\alpha(0)$ and $S^{\CD}_\alpha(+\infty)$ are diagonal
in most of the situations, as easily observed in Proposition \ref{propSurS}, let us define for $a \in \CC$ and $m \in \{0,-1\}$ the following functions:
\begin{equation*}
\varphi_m(\cdot,a):= \varphi_m^-(\cdot)+ a\;\! \tilde{\varphi}_m(\cdot)\ .
\end{equation*}
Then, by a simple computation one obtains
\begin{eqnarray*}
\varphi_m(x,a) &=&
\frac{\Gamma\big(\frac{1}{2}(|m|+1+ix)\big)}{\Gamma\big(\frac{1}{2}(|m|+1-ix)\big)}
\;\!
\frac{\Gamma\big(\frac{1}{2}(|m+\alpha|+1-ix)\big)}{\Gamma\big(\frac{1}{2}(|m+\alpha|+1+ix)\big)}
\ \cdot \\
&& \cdot \ \Big[
e^{i\delta_m^\alpha} + a \;\!e^{-i\pi|m|/2} \;\!
\frac{e^{\pi x/2}}{2\sin\big(\frac{\pi}{2}(1+|m+\alpha|+ix)\big)}\;\! \Big].
\end{eqnarray*}
Let us mention that the equality
\begin{equation}\label{relGamma}
\Gamma(z) \;\!\Gamma\big(1-z)=\frac{\pi}{\sin(\pi z)}
\end{equation}
for $z=\ud(1+|m+\alpha|+ix)$ has been used for this calculation.
In the case $a=0$, the function $\varphi_m(\cdot,0)$ clearly takes its values in $\T$.
We shall now consider the other two special cases $\varphi_0(\cdot,e^{i\pi\alpha}-e^{-i\pi\alpha})$
and $\varphi_{-1}(\cdot,e^{-i\pi\alpha}-e^{i\pi\alpha})$ which will appear naturally subsequently.
Few more calculations involving some trigonometric relations and
the same relation \eqref{relGamma} lead to
\begin{eqnarray*}
\varphi_0(x,e^{i\pi\alpha}-e^{-i\pi\alpha})&=&e^{i\pi\alpha/2}\;\!
\frac{\Gamma\big(\frac{1}{2}(1+ix)\big)}{\Gamma\big(\frac{1}{2}(1-ix)\big)}
\;\!
\frac{\Gamma\big(\frac{1}{2}(1+\alpha-ix)\big)}{\Gamma\big(\frac{1}{2}(1+\alpha+ix)\big)}
\;\!\frac{\sin\big(\frac{\pi}{2}(1+\alpha-ix)\big)}{\sin\big(\frac{\pi}{2}(1+\alpha+ix)\big)}\\
&=&
e^{i\pi\alpha/2}\;\!
\frac{\Gamma\big(\frac{1}{2}(1+ix)\big)}{\Gamma\big(\frac{1}{2}(1-ix)\big)}
\;\!
\frac{\Gamma\big(\frac{1}{2}(1-\alpha-ix)\big)}{\Gamma\big(\frac{1}{2}(1-\alpha+ix)\big)}
\end{eqnarray*}
and to
\begin{eqnarray*}
\varphi_{-1}(x,e^{-i\pi\alpha}-e^{i\pi\alpha}) &=&-e^{-i\pi\alpha/2}\;\!
\frac{\Gamma\big(1+\frac{1}{2}ix\big)}{\Gamma\big(1-\frac{1}{2}ix\big)}
\;\!
\frac{\Gamma\big(1-\frac{1}{2}(\alpha+ix)\big)}{\Gamma\big(1-\frac{1}{2}(\alpha-ix)\big)}
\;\!
\frac{\sin\big(\frac{\pi}{2}(\alpha+ix)\big)}{\sin\big(\frac{\pi}{2}(\alpha-ix)\big)} \\
&=&-e^{-i\pi\alpha/2}\;\!
\frac{\Gamma\big(1+\frac{1}{2}ix\big)}{\Gamma\big(1-\frac{1}{2}ix\big)}
\;\!
\frac{\Gamma\big(\frac{1}{2}(\alpha-ix)\big)}{\Gamma\big(\frac{1}{2}(\alpha+ix)\big)}\ .
\end{eqnarray*}
Clearly, both functions are continuous and take values in $\T$. Furthermore, since
$\varphi^-_m$ and $\tilde{\varphi}_m$ have limits at $\pm \infty$, so does the functions
$\varphi_m(\cdot,a)$. It follows that the variation of the arguments of the previous functions
can be defined. More generally, for any continuously differentiable function $\varphi:[-\infty,\infty]\to \T$ we set
\begin{equation*}
\Var[\varphi]:=\frac{1}{i}\int_{-\infty}^\infty \varphi(x)^{-1}\;\!\varphi'(x)\;\! \dd x\ .
\end{equation*}

Let us first state a convenient formula. Its proof is given in the Appendix~\ref{appb}.

\begin{lemma}\label{variationarg}
Let $a,b>0$. For $\varphi_{a,b}(x):=\frac{\Gamma(a+ix)}{\Gamma(a-ix)}\frac{\Gamma(b-ix)}{
\Gamma(b+ix)}$ one has $\Var[\varphi_{a,b}]=2\pi(a-b)$.
\end{lemma}

As an easy corollary one obtains

\begin{corol}
The following equalities hold:
\begin{enumerate}
\item[i)] $\Var[\varphi_m(\cdot,0)]=2\delta_m^\alpha$ for $m \in \{0,-1\}$,
\item[ii)] $\Var[\varphi_0(\cdot,e^{i\pi\alpha}-e^{-i\pi\alpha})]=\pi\alpha$,
\item[iii)] $\Var[\varphi_{-1}(\cdot,e^{-i\pi\alpha}-e^{i\pi\alpha})]=\pi(2-\alpha)$.
\end{enumerate}
\end{corol}

Let us now set
$$
\phi_1(C,D,\alpha):=\Var\big[\det\big(\Gamma_1(C,D,\alpha,\cdot)\big)\big]
$$
and
$$\phi_3(C,D,\alpha):=-\Var\big[\det\big(\Gamma_3(C,D,\alpha,\cdot)\big)\big].$$
The sign $"-"$ in the second definition comes from the sense of the computation of the winding
number: from $+\infty$ to $-\infty$.
By taking into account the above information and the expression $S^{\CD}_\alpha(0)$ and $S^{\CD}_\alpha(+\infty)$ recalled in Proposition \ref{propSurS}
one can prove:

\begin{prop}
One has
\begin{enumerate}
\item[i)] If $D=0$, then $\phi_3(C,D,\alpha)=0$,

\item[ii)] If $\det(D)\neq 0$, then $\phi_3(C,D,\alpha)=-2\pi$,

\item[iii)] If $\ker(D)= \left(\begin{smallmatrix}
\CC\\ 0 \end{smallmatrix}\right)$ or if
$\dim[\ker(D)]=1$, $\alpha < 1/2$ and $\ker(D)\neq \left(\begin{smallmatrix}
0\\ \CC
\end{smallmatrix}\right)$,
then $\phi_3(C,D,\alpha)=-2\pi(1-\alpha)$,

\item[iv)] If $\ker(D)= \left(\begin{smallmatrix}
0\\ \CC
\end{smallmatrix}\right)$ or if
$\dim[\ker(D)]=1$, $\alpha > 1/2$ and $\ker(D)\neq \left(\begin{smallmatrix}
\CC\\ 0
\end{smallmatrix}\right)$,
then $\phi_3(C,D,\alpha)=-2\pi\alpha$,

\item[v)] If
$\dim[\ker(D)]=1$ and $\alpha =1/2$, then $\phi_3(C,D,\alpha)=-\pi$.
\end{enumerate}
Furthermore,
\begin{enumerate}
\item[a)] If  $C=0$, then $\phi_1(C,D,\alpha)=2\pi$,

\item[b)] If  $\det(C)\ne 0$, then $\phi_1(C,D,\alpha)=0$,

\item[c)] If $\ker (C)=\left(\begin{smallmatrix}0\\ \CC \end{smallmatrix}\right)$ or if  $\dim[\ker(C)]=1$,
$\alpha>1/2$ and
$\ker (C)\ne \left(\begin{smallmatrix}\CC \\ 0 \end{smallmatrix}\right)$,
then  $\phi_1(C,D,\alpha)=2\pi(1-\alpha)$,

\item[d)] If $\ker (C)=\left(\begin{smallmatrix} \CC \\ 0\end{smallmatrix}\right)$ or if
$\dim[\ker (C)]=1$, $\alpha<1/2$ and $\ker (C)\ne \left(\begin{smallmatrix}0 \\ \CC \end{smallmatrix}\right)$,
then $\phi_1(C,D,\alpha)=2\pi\alpha$,

\item[e)] If $\dim[\ker (C)]=1$ and $\alpha=1/2$, then
$\phi_1(C,D,\alpha)=\pi$.
\end{enumerate}
\end{prop}

\begin{proof}
Statements i) to iv) as well as statements a) to d) are easily obtained simply by
taking the asymptotic values of $S^{\CD}_\alpha(\cdot)$ into account. So let us concentrate
on the remaining statements.

Let $p=(p_1,p_2)\in \C^2$ with $\|p\|=1$, and let
\[
\P=\begin{pmatrix}
|p_2|^2 & - p_1 \Bar p_2\\
-\Bar p_1 p_2 & |p_1|^2
\end{pmatrix}
\]
be the orthogonal projection onto $p^\bot$. For $x \in \R$, let us also set
\begin{equation*}
\varphi(\P,x):=
\Big(
\begin{smallmatrix}\varphi^-_0(x) & 0 \\ 0 & \varphi^-_{-1}(x) \end{smallmatrix}\Big) +
\Big(
\begin{smallmatrix}\tilde{\varphi}_0(x) & 0 \\ 0 & \tilde{\varphi}_{-1}(x) \end{smallmatrix}\Big)
2\P
\Big(\begin{smallmatrix} i & 0 \\ 0 & -i \end{smallmatrix}\Big)
\end{equation*}
whose determinant is equal to
\begin{equation*}
g(x):= \varphi^-_0(x)\;\!\varphi^-_{-1}(x) + 2i\tilde{\varphi}_0(x)\;\!
\varphi^-_{-1}(x)\;\!|p_2|^2
-2i\varphi^-_0(x)\;\!\tilde{\varphi}_{-1}(x)|p_1|^2\ .
\end{equation*}
By taking the explicit expressions for these functions one obtains
\begin{eqnarray*}
g(x)&=& \frac{\Gamma\big(\frac{1}{2}(1+ix)\big)}{\Gamma\big(\frac{1}{2}(1-ix)\big)}
\;\!
\frac{\Gamma\big(\frac{1}{2}(\frac{3}{2}-ix)\big)}{\Gamma\big(\frac{1}{2}(\frac{3}{2}+ix)\big)}\;\!
\frac{\Gamma\big(\frac{1}{2}(2+ix)\big)}{\Gamma\big(\frac{1}{2}(2-ix)\big)}
\frac{\Gamma\big(\frac{1}{2}(\frac{3}{2}-ix)\big)}{\Gamma\big(\frac{1}{2}(\frac{3}{2}+ix)\big)}\;\!
\\
&&\cdot\, \Big(
1+ i\;\!e^{i\pi/4}\;\!\frac{e^{\pi x /2}}{\pi}\;\!
\Gamma\big(\ud(\ud-ix)\big)\Gamma\big(\ud({\textstyle\frac{3}{2}}+ix)\big)
\Big).
\end{eqnarray*}
Now, by setting $z=\frac{3}{4}+i\frac{x}{2}$ and by some algebraic computations one obtains
\begin{eqnarray*}
&&1+ie^{i\pi/4}\;\!\frac{e^{\pi x /2}}{\pi}\;\!
\Gamma\big(\ud(\ud-ix)\big)\;\!\Gamma\big(\ud({\textstyle\frac{3}{2}}+ix)\big)\\
&=& 1+\frac{i}{\pi}e^{-i\pi(z-1)}\;\!\Gamma(1-z)\;\!\Gamma(z)
=1-i\frac{e^{-i\pi z}}{\sin(\pi z)} \\
&=&-i\frac{\cos(\pi z)}{\sin(\pi z)}
= -i\frac{1}{\tan\big(\frac{3\pi}{4}+i\frac{\pi x}{2}\big)}\\
&=&-i\frac{\tanh(\frac{\pi x}{2})-i}{\tanh(\frac{\pi x)}{2})+i}\ .
\end{eqnarray*}
Thus, one finally obtains that
\begin{equation*}
g(x)=-i\frac{\Gamma\big(\frac{1}{2}(1+ix)\big)}{\Gamma\big(\frac{1}{2}(1-ix)\big)}
\;\!
\frac{\Gamma\big(\frac{1}{2}(\frac{3}{2}-ix)\big)}{\Gamma\big(\frac{1}{2}(\frac{3}{2}+ix)\big)}\;\!
\frac{\Gamma\big(\frac{1}{2}(2+ix)\big)}{\Gamma\big(\frac{1}{2}(2-ix)\big)}\;\!
\frac{\Gamma\big(\frac{1}{2}(\frac{3}{2}-ix)\big)}{\Gamma\big(\frac{1}{2}(\frac{3}{2}+ix)\big)}
\frac{\tanh(\frac{\pi x}{2})-i}{\tanh(\frac{\pi x}{2})+i}\ .
\end{equation*}
Note that this function does not depend on the projection $\P$ at all.

Clearly one has
\[
\Var[g]=\Var[\varphi_{\frac{1}{2},\frac{3}{4}}\,] +\Var[\varphi_{1,\frac{3}{4}}\,]
+\Var\Big[ \frac{\tanh(\frac{\pi \cdot}{2})-i}{\tanh(\frac{\pi \cdot}{2})+i}\,\Big]=
-\dfrac{\pi}{2}+\dfrac{\pi}{2}+\pi=\pi.
\]
Now, by observing that $\phi_3(C,D,\alpha)=-\Var[g]$ in the case v), one concludes
that in this special case $\phi_3(C,D,\alpha)=-\pi$.

For the case e), observe that by setting $\P:=1-\Pi$, one easily obtains that
in this special case $\Gamma_1(C,D,\alpha,\cdot)=\varphi(\P,\cdot)$. It follows
that $\phi_1(C,D,\alpha)=\Var[g]$ and then $\phi_1(C,D,\alpha)=\pi$.
\end{proof}

\subsection{Contribution of $\Gamma_2(C,D,\alpha,\cdot)$}

Recall first that $\Gamma_2(C,D,\alpha,\cdot)$ defined in \eqref{Gma2} is equal to $S^{\CD}_\alpha(\cdot)$.
We are interested here in the phase of $\det \big(S^{\CD}_\alpha(\kappa)\big)$ acquired
as $\kappa$ runs from $0$ to $+\infty$; we denote this phase by $\phi_2(C,D,\alpha)$.
Note that if $\det \big(S^{\CD}_\alpha(\kappa)\big)=\frac{\bar f(\kappa)}{f(\kappa)}$ for a non-vanishing continuous function $f:\R_+ \to \C^*$, then
\begin{equation*}
\phi_2(C,D,\alpha)=-2\big(\arg f(+\infty)-\arg f(0)\big)\ ,
\end{equation*}
where $\arg:\R_+ \to \R$ is a continuous function defined by the argument of $f$.
In the sequel, we shall also use the notation $\Arg:\C^*\to (-\pi,\pi]$ for the principal argument of a complex number different from $0$.

Now, let us consider $\kappa>0$ and set $S(\kappa):=S^{\CD}_\alpha(\kappa)$.
For shortness, we also set $L:=\frac{\pi}{2\sin(\pi\alpha)}\;\!C$ and
\begin{gather*}
B:=\left(\begin{smallmatrix}
b_1(\kappa) & 0 \\
0 & b_2(\kappa)
\end{smallmatrix}\right)=\left(\begin{smallmatrix}
\frac{\Gamma(1-\alpha)}{2^\alpha}\;\!\kappa^{\alpha} & 0 \\
0 & \frac{ \Gamma(\alpha)}{2^{1-\alpha}}\;\!\kappa^{(1-\alpha)}
\end{smallmatrix}\right),
\quad
\Phi:=\left(\begin{smallmatrix}
e^{-i\pi\alpha/2} & 0 \\
0 & e^{-i\pi(1-\alpha)/2}
\end{smallmatrix}\right),\quad
J:=\left(\begin{smallmatrix}
1 & 0 \\ 0 & -1
\end{smallmatrix}\right).
\end{gather*}
Note that the matrices $B$, $\Phi$ and $J$ commute with each other, that the matrix $B$
is self-adjoint and invertible, and that $J$ and $\Phi$ are unitary.

I) If $D=0$, then $S^{\CD}_\alpha$ is constant and $\phi_2(C,D,\alpha)=0$.

II) Let us assume $\det(D)\neq 0$, {\it i.e.}~$D$ is invertible.
Without loss of generality, we may assume that $D=1$, as explained in \cite[Sec.~3]{PR}, and that $C$ and hence $L$
are self-adjoint. We write $C=(c_{jk})$, $L=(l_{jk})$ and we then use the expression
\begin{equation}\label{eq-SL}
S(\kappa)=\Phi \;\!\frac{B^{-1}\;\! L\;\!B^{-1} +\cos(\pi\alpha)J
+i\sin(\pi\alpha)}{B^{-1}\;\! L\;\!B^{-1} +\cos(\pi\alpha)J -i\sin(\pi\alpha)}
\;\!\Phi\;\!J\ ,
\end{equation}
derived in \cite{PR}. By direct calculation one obtains
$\det \big(S(\kappa)\big)=  \frac{\bar f(\kappa)}{f(\kappa)}$
with
\begin{eqnarray}\label{f1}
\nonumber f(\kappa) &=&
\det\big(B^{-1} L B^{-1} +\cos(\pi\alpha)J -i\sin(\pi\alpha)\big) \\
\nonumber &=&\det(L)\;\! b_1^{-2}(\kappa) \;\!b_2^{-2}(\kappa) -1 +\cos(\pi\alpha) \big(l_{22}\;\!
b_2^{-2}(\kappa) - l_{11} \;\!b_1^{-2}(\kappa)\big)\\
&&-i\sin(\pi\alpha) \big(l_{11} \;\!b_1^{-2}(\kappa)+l_{22} \;\!b_2^{-2}(\kappa)\big)
\end{eqnarray}
and $f$ is non-vanishing as the determinant of an invertible matrix.

For the computation of $\phi_2(C,D,\alpha)$ we shall have to consider several cases. We first assume that $\det(C)\neq 0$, which is equivalent to $\det(L)\neq 0$. In that case one clearly has $\det\big(S^{\CD}_\alpha(0)\big) = \det\big(S^{\CD}_\alpha(+\infty)\big)$, and then $\phi_2(C,D,\alpha)$ will be a multiple of $2\pi$. Furthermore, note that $\Arg \big(f (+\infty)\big)= \pi$ and that $\Arg\big( f(0)\big)=0$ if $\det(L)>0$ and $\Arg\big( f(0)\big)=\pi$ if $\det(L)<0$.

Assuming that $l_{11}\;\!l_{22}\geq 0$ (which means that $\Im f$ is either non-negative or non-positive, and its sign is opposite to that of $\tr(L)$), one has the following cases:
\begin{enumerate}
\item[II.1)] If  $\tr (C)>0$ and $\det(C)>0$, then $\Im f<0$ and
$\phi_2(C,D,\alpha)=2\pi$,
\item[II.2)] If $\tr (C)>0$ and $\det(C)<0$, then $\Im f<0$ and
$\phi_2(C,D,\alpha)=0$,
\item[II.3)] If $\tr (C)<0$ and $\det(C)>0$, then $\Im f>0$ and
$\phi_2(C,D,\alpha)=-2\pi$,
\item[II.4)] If $\tr (C)<0$ and $\det(C)<0$, then $\Im f>0$ and
$\phi_2(C,D,\alpha)=0$,
\item[II.5)] If $c_{11}=c_{22}=0$ (automatically $\det(C)<0$), then $f$ is real and non-vanishing, hence $\phi_2(C,D,\alpha)=0$.
\end{enumerate}

Now, if $l_{11} l_{22}<0$ the main difference is that the parameter $\alpha$
has to be taken into account. On the other hand, one has $\det(L)<0$ which implies that $\arg f(+\infty)-\arg f(0)$ has to be a multiple of $2\pi$. For the computation of this difference, observe that the equation $\Im f(\kappa)=0$ (for $\kappa\ge 0$) is equivalent to
\begin{equation}\label{eq-kkk}
\dfrac{b_1^{-2}(\kappa)}{b_2^{-2}(\kappa)}=-\dfrac{l_{22}}{l_{11}}\Longleftrightarrow
\kappa^{2\alpha-1} = 2^{2\alpha-1} \dfrac{\Gamma(\alpha)}{\Gamma(1-\alpha)}\sqrt{-\dfrac{l_{11}}{l_{22}}}.
\end{equation}
For $\alpha\ne 1/2$ this equation has a unique solution $\kappa_0$,  and it follows that the sign of $\Im f(\kappa)$ will be different for $\kappa<\kappa_0$ and for $\kappa>\kappa_0$ (and will depend on $\alpha$ and on the relative sign of $l_{11}$ and $l_{22}$).

Let us now estimate $\Re f(\kappa_0)$. We have
\begin{eqnarray*}
\Re f(\kappa)&=&\det (L)\;\! b_1^{-2}(\kappa)\;\! b_2^{-2}(\kappa) -1 +\cos(\pi\alpha) \big(l_{22} \;\!b_2^{-2}(\kappa) - l_{11}
\;\!b_1^{-2}(\kappa)\big)\\
&\leq&  -|l_{11}\;\!l_{22}|\;\! b_1^{-2}(\kappa) \;\!b_2^{-2}(\kappa) -1 +|\cos(\pi\alpha)|\big( \big|l_{22}|\;\! b_2^{-2}(\kappa) + |l_{11}|\;\! b_1^{-2}(\kappa)\big)\\
&=& -\Big(|l_{11}\;\!l_{22}|\;\! b_1^{-2}(\kappa) \;\!b_2^{-2}(\kappa) +1 -|\cos(\pi\alpha)|\big( \big|l_{22}| \;\!b_2^{-2}(\kappa) + |l_{11}| \;\! b_1^{-2}(\kappa)\big)\Big)\\
&=&-\big(1-|\cos(\pi\alpha)|\big)\big(|l_{11}\;\!l_{22}|\;\! b_1^{-2}(\kappa)\;\! b_2^{-2}(\kappa) +1\big)\\
&&-|\cos(\pi\alpha)| \big( |l_{11}|\;\!b_1^{-2}(\kappa)-1\big)\big( |l_{22}\;\!|b_2^{-2}(\kappa)-1\big).
\end{eqnarray*}
Hence using \eqref{eq-kkk} and the equality $-\frac{l_{22}}{l_{11}}= \frac{|l_{22}|}{|l_{11}|}$ one obtains
\[
\Re f(\kappa_0)\le -\big(1-|\cos(\pi\alpha)|\big)\big(|l_{11}l_{22}|\, b_1^{-2}(\kappa_0) b_2^{-2}(\kappa_0) +1\big)
-|\cos(\pi\alpha)| \big( |l_{22}|b_2^{-2}(\kappa_0)-1\big)^2<0.
\]
This estimate implies that $0$ is not contained in the interior of the curve $f(\R_+)$, which means that
$\arg f(+\infty)-\arg f(0)=0$ for all $\alpha \neq 1/2$.

For the special case $\alpha=1/2$, the equation \eqref{eq-kkk} has either no solution or holds for all $\kappa\in \R_+$. In the former situation, $\Im f$ has always the same sign, which means that the $\arg f(+\infty)-\arg f(0)=0$. In the latter situation, $f$ is real, and obviously $\arg f(+\infty)-\arg f(0)=0$.
In summary, one has obtained:
\begin{enumerate}
\item[II.6)] If $c_{11}\;\!c_{22}<0$, then $\phi_2(C,D,\alpha)=0$.
\end{enumerate}

Let us now assume that $\det (C)=0$ but $C\ne 0$, {\it i.e.}~$\det(L)=0$ but $L\neq 0$. In that case one simply has
\[
f(\kappa)= -1 +\cos(\pi\alpha) \big(l_{22}\;\! b_2^{-2}(\kappa) - l_{11}\;\! b_1^{-2}(\kappa)\big)
-i\sin(\pi\alpha) \big(l_{11} \;\!b_1^{-2}(\kappa)+l_{22} \;\!b_2^{-2}(\kappa)\big).
\]
Furthermore, one always has $l_{11}l_{22}\ge 0$, which means
that $\Im f$ is either non-negative or non-positive.
Then, since $\Arg \big(f(+\infty)\big)=\pi$, it will be sufficient
to calculate the value $\Arg \big(f(0)\big)$.

i) Assume first that $l_{11}=0$, which automatically implies that $l_{22}\neq 0$ and $l_{12}=l_{21}=0$. Then one has
\[
f(\kappa)= -1 +\cos(\pi\alpha) \;\!l_{22}\;\! b_2^{-2}(\kappa) -i\sin(\pi\alpha)\;\! l_{22}\;\! b_2^{-2}(\kappa)
\]
and
\[
\Arg \big(f(0)\big)=\begin{cases}
-\pi\alpha & \text{ if } l_{22}>0\\
\pi(1-\alpha) & \text{ if } l_{22}<0
\end{cases}\ .
\]
By taking into account the sign of $\Im f$, one then obtains
\[
\arg f(+\infty)-\arg f(0)=
\begin{cases}
-\pi(1-\alpha) & \text{ if } l_{22}>0\\
\pi \alpha & \text{ if } l_{22}<0
\end{cases}\ .
\]

ii) Similarly, if we assume now that $l_{22}=0$, we then have $l_{11}\ne 0$, $l_{12}=l_{21}=0$ and
\[
f(\kappa)= -1 -\cos(\pi\alpha) \;\!l_{11} \;\!b_1^{-2}(\kappa)
-i\sin(\pi\alpha) \;\!l_{11}\;\! b_1^{-2}(\kappa).
\]
It then follows that
\[
\Arg \big(f(0)\big)=\begin{cases}
\pi\alpha & \text{ if } l_{11}<0\\
-\pi(1-\alpha) & \text{ if } l_{11}>0
\end{cases}
\]
and
\[
\arg f(+\infty)-\arg f(0)=
\begin{cases}
\pi(1-\alpha) & \text{ if } l_{11}<0\\
-\pi \alpha & \text{ if } l_{11}>0
\end{cases}\ .
\]

iii) Assume now that $l_{11}\;\!l_{22}\ne 0$ (which means automatically $l_{11}l_{22}>0$) and that $\alpha=1/2$.
Since $b_1(\kappa)=b_2(\kappa)=:b(\kappa)$ one then easily observes that
$f(\kappa)=-1-i\tr (L)\;\! b^{-2}(\kappa)$,
$\Arg \big(f(0)\big)=-\frac{\pi}{2} \sign \big(\tr (L)\big)$ and
$\arg f(+\infty)-\arg f(0)=-\frac{\pi}{2} \sign \big(\tr (L)\big)$.

iv) Assume that $l_{11}\;\!l_{22}\ne 0$ and that $\alpha<1/2$.
In this case one can rewrite
\[
f(\kappa)= -1 +\cos(\pi\alpha)\;\! b_2^{-2}(\kappa)\Big(l_{22}  - l_{11} \frac{b_2^2(\kappa)}{b_1^2(\kappa)}\Big)
-i\sin(\pi\alpha)\;\!b_2^{-2}(\kappa) \Big(l_{22}+l_{11} \frac{b_2^2(\kappa)}{b_1^2(\kappa)}\Big).
\]
Since $b_2(\kappa)/b_1(\kappa)\to 0$ as $\kappa\to 0+$, one has the same limit values and phases as in i).

v) Similarly, if $l_{11}\;\!l_{22}\ne 0$ and  $\alpha>1/2$, we have the same limit and phases as in ii).

In summary, if $\det(C)=0$ and $C\neq 0$ one has obtained:
\begin{enumerate}
\item[II.7)] If $c_{11}=0$ and $\tr(C)>0$, or if $c_{11}\;\!c_{22}\neq 0$, $\tr(C)>0$ and $\alpha<1/2$, then $\phi_2(C,D,\alpha)=2\pi(1-\alpha)$,
\item[II.8)] If $c_{11}=0$ and $\tr(C)<0$, or if $c_{11}\;\!c_{22}\neq 0$, $\tr(C)<0$ and $\alpha<1/2$, then $\phi_2(C,D,\alpha)=-2\pi\alpha$,
\item[II.9)] If $c_{22}=0$ and $\tr(C)>0$, or if $c_{11}\;\!c_{22}\neq 0$, $\tr(C)>0$ and $\alpha>1/2$, then $\phi_2(C,D,\alpha)=2\pi\alpha$,
\item[II.10)] If $c_{22}=0$ and $\tr(C)<0$, or if $c_{11}\;\!c_{22}\neq 0$, $\tr(C)<0$ and $\alpha>1/2$, then $\phi_2(C,D,\alpha)=-2\pi(1-\alpha)$,
\item[II.11)] If $c_{11}\;\!c_{22}\neq 0$, $\tr(C)>0$ and $\alpha = 1/2$, then $\phi_2(C,D,\alpha)=\pi$,
\item[II.12)] If $c_{11}\;\!c_{22}\neq 0$, $\tr(C)<0$ and $\alpha = 1/2$, then $\phi_2(C,D,\alpha)=-\pi$.
\end{enumerate}

III) If $C=0$, then $S^{\CD}_\alpha$ is constant and $\phi_2(C,D,\alpha)=0$.

IV) We shall now consider the situation $\det(D)=0$ but $D\ne 0$.
Obviously, $\ker(D)$ is of dimension $1$. So let $p=(p_1,p_2)$ be a vector
in $\ker(D)$ with $\|p\|=1$.
Let us also introduce
\begin{equation*}
c(\kappa)=b_1^2(\kappa)\;\! |p_2|^2\;\!e^{-i\pi\alpha} - b_2^2(\kappa)\;\! |p_1|^2\;\!e^{i\pi\alpha}
\end{equation*}
and
\begin{equation*}
X_-:= \big(b_1^2(\kappa)\;\!|p_2|^2 - b_2^2(\kappa)\;\!|p_1|^2\big),\qquad
X_+:= \big(b_1^2(\kappa)\;\!|p_2|^2 + b_2^2(\kappa)\;\!|p_1|^2\big)\ .
\end{equation*}
In that case it has been shown in \cite{PR} that
\begin{equation*}
S= \Phi\;\!\big(c(\kappa)+\ell\big)^{-1} M(\kappa)\;\!\Phi\;\!J\ ,
\end{equation*}
where
\[
M(\kappa):=
\begin{pmatrix}
e^{i\pi\alpha}\;\!X_- + \ell & -2\;\!i\;\!\sin(\pi\alpha)\;\!b_1(\kappa)\;\!b_2(\kappa)\;\!p_1\;\!\bar p_2 \\
-2\;\!i\;\!\sin(\pi\alpha)\;\!b_1(\kappa)\;\!b_2(\kappa)\;\!\bar p_1\;\! p_2 &
e^{-i\pi\alpha}\;\!X_- + \ell
\end{pmatrix}
\]
and $\ell$ is a real number which will be specified below.
Note that $\det \big(M(\kappa)\big)=|c(\kappa)+\ell|^2$ which ensures that $S$ is a unitary operator.
Therefore, by setting
\[
g(\kappa):=c(\kappa)+\ell\\
=\cos(\pi\alpha)\Big(
b_1^2(\kappa)\;\!|p_2|^2-b_2^2(\kappa)\;\!|p_1|^2
\Big)+\ell
-i\sin(\pi\alpha)
\Big(
b_1^2(\kappa)\;\!|p_2|^2+b_2^2(\kappa)\;\!|p_1|^2
\Big),
\]
one has
\[
\phi_2(C,D,\alpha)=-2\big(\arg g(+\infty)- \arg g(0)\big),
\]
where $\arg:\R_+ \to \R$ is a continuous function defined by the argument of $g$.
Note already that we always have $\Im g< 0$.

We first consider the special case $\alpha=1/2$. In that case we  have $b_1(\kappa)=b_2(\kappa)=:b(\kappa)$, and then
\[
g(\kappa)=\ell-i b^2(\kappa).
\]
If $\ell\ne 0$, we have $\Arg \big(g(0)\big)=\Arg(\ell)$ and $\Arg \big(g(+\infty)\big)=-\pi/2$.
Therefore
\[
\arg g(+\infty)-\arg g(0)=\begin{cases}
-\pi/2 & \text{if } \ell>0\\
\pi/2 & \text{if } \ell<0
\end{cases} \ .
\]
If $\ell=0$, then $g$ is pure imaginary, hence $\arg g(+\infty)-\arg g(0)=0$.
In summary, for $\det(D)=0$ but $D \neq 0$, one has already obtained:
\begin{enumerate}
\item[IV.1)] If $\ell>0$ and $\alpha=1/2$, then $\phi_2(C,D,\alpha)=\pi$,
\item[IV.2)] If $\ell=0$ and $\alpha=1/2$, then $\phi_2(C,D,\alpha)=0$,
\item[IV.3)] If $\ell<0$ and $\alpha=1/2$, then $\phi_2(C,D,\alpha)=-\pi$.
\end{enumerate}

Let us now consider the case $\alpha<1/2$, and assume first that  $\ell \ne 0$. It follows that $\Arg \big(g(0)\big)=\Arg (\ell)$. To calculate $\Arg \big(g(+\infty)\big)$ one has to consider two subcases.
So, on the one hand let us assume in addition that $p_1\ne 0$. Then one has
\[
g(\kappa)= \ell -
\cos(\pi\alpha)\;\!b_2^2(\kappa)\Big(|p_1|^2-\dfrac{b_1^2(\kappa)}{b_2^2(\kappa)}\;\!|p_2|^2\Big)
-i\sin(\pi\alpha)\;\!b_2^2(\kappa)\Big(|p_1|^2+\dfrac{b_1^2(\kappa)}{b_2^2(\kappa)}\;\!|p_2|^2\Big)
.
\]
Since $b_1(\kappa)/b_2(\kappa)\to 0$ as $\kappa\to +\infty$,
one obtains $\Arg \big(g(+\infty)\big)=-\pi(1-\alpha)$
and
\[
\arg g(+\infty)-\arg g(0)=\begin{cases}
-\pi(1-\alpha), & \text{if } \ell>0\\
\pi\alpha, & \text{if } \ell<0
\end{cases} \ .
\]
On the other hand, if $p_1=0$, then one has
\[
g(\kappa)=\ell + b_1^2(\kappa)\big( \cos(\pi\alpha) -i\sin(\pi\alpha)\big),
\]
which implies that $\Arg \big(g(+\infty)\big)=-\pi\alpha$ and that
\[
\arg g(+\infty)-\arg g(0)=\begin{cases}
-\pi\alpha, & \text{if } \ell>0\\
\pi(1-\alpha), & \text{if } \ell<0
\end{cases}\ .
\]

Now, let us assume that $\ell=0$. In this case the above limits for $\kappa \to + \infty$ still hold, so we only need
to calculate $\Arg \big(g(0)\big)$. Firstly, if $p_2\ne 0$, we have
\[
g(\kappa)= \cos(\pi\alpha)\;\!b_1^2(\kappa)\Big( |p_2|^2-\dfrac{b_2^2(\kappa)}{b_1^2(\kappa)}\;\!|p_1|^2\Big)
-i\sin(\pi\alpha)\;\!b_1^2(\kappa) \Big(|p_2|^2+\dfrac{b_2^2(\kappa)}{b_1^2(\kappa)}\;\!|p_1|^2 \Big),
\]
and since $b_2(\kappa)/b_1(\kappa)\to 0$ as $\kappa \to 0+$ it follows that $\Arg \big(g(0)\big)=-\pi\alpha$. Secondly, if $p_2=0$, then
\[
g(\kappa)=-b_2^2(\kappa)\Big(\cos(\pi\alpha) +i\sin(\pi\alpha)\big)
\Big),
\]
and we get $\Arg \big(g(0)\big)=-\pi(1-\alpha)$.

In summary, for $\det(D)=0$, $D\neq 0$ and $\alpha<1/2$, we have
obtained
\begin{enumerate}
\item[IV.4)] if $\ell <0$ and $p_1\neq 0$, then $\phi_2(C,D,\alpha)=-2\pi\alpha$,
\item[IV.5)] if $\ell<0 $ and $p_1= 0$, then $\phi_2(C,D,\alpha)=-2\pi(1-\alpha)$,
\item[IV.6)] if $\ell>0 $ and $p_1\neq 0$, then $\phi_2(C,D,\alpha)=2\pi(1-\alpha)$,
\item[IV.7)] if $\ell >0$ and $p_1= 0$, then $\phi_2(C,D,\alpha)=2\pi\alpha$,
\item[IV.8)] if $\ell=0 $, $p_1\neq 0$ and $p_2\neq 0$ then $\phi_2(C,D,\alpha)=2\pi(1-2\alpha)$,
\item[IV.9)] if $\ell=0$ and $p_1= 0$ or if
$\ell=0$ and $p_2=0$, then $\phi_2(C,D,\alpha)=0$.
\end{enumerate}

The case $\det(D)=0$, $D\neq 0$ and $\alpha>1/2$ can be treated analogously. We simply state the results:
\begin{enumerate}
\item[IV.10)] if $\ell <0$ and $p_2\neq 0$, then $\phi_2(C,D,\alpha)=-2\pi(1-\alpha)$,
\item[IV.11)] if $\ell<0 $ and $p_2= 0$, then $\phi_2(C,D,\alpha)=-2\pi\alpha$,
\item[IV.12)] if $\ell>0 $ and $p_2\neq 0$, then $\phi_2(C,D,\alpha)=2\pi\alpha$,
\item[IV.13)] if $\ell >0$ and $p_2= 0$, then $\phi_2(C,D,\alpha)=2\pi(1-\alpha)$,
\item[IV.14)] if $\ell=0 $, $p_1\neq 0$ and $p_2\neq 0$ then $\phi_2(C,D,\alpha)=-2\pi(1-2\alpha)$,
\item[IV.15)] if $\ell=0$ and $p_1= 0$ or if
$\ell=0$ and $p_2=0$, then $\phi_2(C,D,\alpha)=0$.
\end{enumerate}

Let us finally recall some relationship between the constant $\ell$ and the matrices $C$ and $D$ in the case IV).
As explained before, we can always assume that $C=(1-U)/2$ and $D=i(1+U)/2$ for some $U\in U(2)$. Recall that in deriving the equalities (IV.1)--(IV.15)
we assumed $\dim[\ker(D)]=1$, {\it i.e.}~$-1$ is an eigenvalue of $U$ of multiplicity $1$.
Let $e^{i\theta}$, $\theta\in(-\pi,\pi)$ be the other eigenvalue of $U$.
Then by the construction explained in \cite{PR}, one has
\[
\ell = \dfrac{\pi}{2\sin(\pi\alpha)} \,\dfrac{1-e^{i\theta}}{i(1+e^{i\theta})}=
- \dfrac{\pi}{2\sin(\pi\alpha)} \frac{\sin \big(\frac{\theta}{2}\big)}{\cos \big(\frac{\theta}{2}\big)}.
\]
On the other hand, the eigenvalues of the matrix $CD^*=i(U-U^*)/4$
are $\lambda_1=0$ and
\[
\lambda_2=i(e^{i\theta}-e^{-i\theta})/4=-\ud \sin(\theta)=
- \sin\big({\textstyle \frac{\theta}{2}}\big)\;\!\cos\big( {\textstyle \frac{\theta}{2}}\big)\ .
\]
It follows that $\lambda_2$ and $\ell$ have the same sign.
Therefore, in (IV.1)--(IV.15) one has:
$\ell <0$ if $CD^*$ has one zero eigenvalue and one negative eigenvalue,
$\ell =0$ if $CD^*=0$ and $\ell>0$ if $CD^*$ has one zero eigenvalue and one positive eigenvalue.

\subsection{Case-by-case results}\label{Sectionfinal}

In this section we finally collect all previous results and prove the case-by-case version of Levinson's theorem. The interest of this analysis is that the contribution of the $0$-energy operator $\Gamma_1(C,D,\alpha,\cdot)$ and the contribution of the operator $\Gamma_3(C,D,\alpha,\cdot)$ at $+\infty$-energy are explicit.
Here, Levinson's theorem corresponds to the equality between the number of bound states of $H_\alpha^{\CD}$ and $-\frac{1}{2\pi} \sum_{j=1}^4\phi_j(C,D,\alpha)$. This is proved again by comparing the column $3$ with the column $7$ (the contribution of $\Gamma_4(C,D,\alpha,\cdot)$ defined in \eqref{Gma4} is always trivial).

For simplicity, we shall write $H$ for $H_\alpha^{\CD}$ and $\phi_j$ for $\phi_j(C,D,\alpha)$. We also recall that the number $\#\sigma_p(H)$ of eigenvalues of $H$ is equal to the number of negative eigenvalues of the matrix $CD^*$ \cite[Lem.~4]{PR}.

We consider first the very special situations:
\begin{center}
\begin{tabular}{|c|c|c|c|c|c|c|}
\hline
 No & Conditions & $\#\sigma_p(H)$ & $\phi_1$ & $\phi_2$ & $\phi_3$  & $\sum_j\phi_j$
 \\ \hline\hline
I &$ D=0 $&$ 0 $&$ 0 $&$ 0 $&$ 0 $&$ 0 $\\\hline
III & $C=0 $&$ 0 $&$ 2\pi $&$ 0 $&$
 -2\pi $&$ 0$ \\ \hline
\end{tabular}
\end{center}

Now, if $\det(D)\neq 0$ and $\det(C)\ne 0$, we set $E:=D^{-1}C=:(e_{jk})$ and obtains:
\begin{center}
\begin{tabular}{|c|c|c|c|c|c|c|}
\hline
 No & Conditions & $\#\sigma_p(H)$ & $\phi_1$ & $\phi_2$ & $\phi_3$  & $\sum_j\phi_j$
 \\ \hline\hline
II.1 &$ e_{11} e_{22}\ge 0$, $\tr(E)>0$, $\det(E)>0$ &$ 0 $&$ 0 $&$ 2\pi $&$ -2\pi $&$ 0$ \\\hline
II.2 &$ e_{11} e_{22}\ge 0$, $\tr(E)>0$, $\det(E)<0$ &$ 1$ &$ 0 $&$ 0 $&$ -2\pi $&$ -2\pi $\\\hline
II.3 &$ e_{11} e_{22}\ge 0$, $\tr(E)<0$, $\det(E)>0$ &$ 2 $&$ 0 $&$ -2\pi $&$ -2\pi $&$ -4\pi$ \\\hline
II.4 &$ e_{11} e_{22}\ge 0$, $\tr(E)<0$, $\det(E)<0$ &$ 1 $&$ 0 $&$ 0 $&$ -2\pi $&$ -2\pi$ \\\hline
II.5 &$ e_{11}=e_{22}=0,\det(E)< 0$&$ 1 $&$ 0 $&$ 0 $&$ -2\pi $&$ -2\pi$ \\\hline
II.6 &$ e_{11}\;\! e_{22}<0 $&$ 1 $&$ 0 $&$ 0 $&$ -2\pi $&$ -2\pi$ \\\hline
\end{tabular}
\end{center}

If $\det(D)\neq 0$, $\det(C)=0$ and if we still set $E:=D^{-1}C$ one has:
\begin{center}
\begin{tabular}{|c|c|c|c|c|c|c|}
\hline
 No & Conditions & $\#\sigma_p(H)$ & $\phi_1$ & $\phi_2$ & $\phi_3$  & $\sum_j\phi_j$
 \\ \hline\hline
II.7.a &$ e_{11}=0, \tr(E)>0 $&$ 0 $&$ 2\pi\alpha $&
$ 2\pi(1-\alpha) $&$ -2\pi $&$ 0 $\\\hline
II.7.b &$e_{11}\;\!e_{22}\neq 0,\tr(E)>0,\alpha<1/2 $&$ 0 $&$ 2\pi\alpha $&
$ 2\pi(1-\alpha) $&$ -2\pi $&$ 0 $\\\hline
II.8.a &$ e_{11}>0, \tr(E)<0 $&$ 1 $&$ 2\pi\alpha $&$ -2\pi\alpha $&$ -2\pi $&$ -2\pi$ \\\hline
II.8.b &$ e_{11}\;\!e_{22}\neq0, \tr(E)<0,\alpha<1/2 $&$ 1 $&$ 2\pi\alpha $&$ -2\pi\alpha $&$ -2\pi $&
$ -2\pi$ \\\hline
II.9.a &$ e_{22}=0, \tr(E)>0 $&$ 0 $&$ 2\pi(1-\alpha)$&$ 2\pi\alpha $&$ -2\pi $&$ 0$ \\\hline
II.9.b &$ e_{11}\;\!e_{22}\neq0, \tr(E)>0,\alpha>1/2 $&$ 0 $&$ 2\pi(1-\alpha)$&$ 2\pi\alpha $&
$ -2\pi $&$ 0$ \\\hline
II.10.a &$ e_{22}=0, \tr(E)<0 $&$ 1 $&$ 2\pi(1-\alpha) $&$ -2\pi(1-\alpha) $&$ -2\pi $&$ -2\pi $\\\hline
II.10.b &$ e_{11}\;\!e_{22}\neq0, \tr(E)<0,\alpha>1/2 $&$ 1 $&$ 2\pi(1-\alpha) $&$ -2\pi(1-\alpha) $&
$ -2\pi $&$ -2\pi $\\\hline
II.11 &$ e_{11}\;\!e_{22}\neq0,\tr(E)> 0 ,\alpha=1/2$&$ 0 $&$ \pi $&$ \pi $&$ -2\pi $&$ 0$ \\\hline
II.12 &$ e_{11}\;\! e_{22}\neq 0,\tr(E)<0,\alpha=1/2 $&$ 1 $&$ \pi $&$ -\pi $&$ -2\pi $&$ -2\pi$ \\\hline
\end{tabular}
\end{center}

On the other hand, if $\dim[\ker(D)]=1$ and $\alpha=1/2$ one has:
\begin{center}
\begin{tabular}{|c|c|c|c|c|c|c|}
\hline
 No & Conditions & $\#\sigma_p(H)$ & $\phi_1$ & $\phi_2$ & $\phi_3$  & $\sum_j\phi_j$
 \\ \hline\hline
IV.1 &$ \ell>0 $&$ 0 $&$ 0 $&$ \pi $&$ -\pi $&$ 0$ \\\hline
IV.2 &$ \ell=0 $&$ 0 $&$ \pi $&$ 0 $&$ -\pi $&$ 0$ \\\hline
IV.3 &$ \ell<0 $&$ 1 $&$ 0 $&$ -\pi $&$ -\pi$&$ -2\pi$ \\\hline
\end{tabular}
\end{center}

If $\dim[\ker(D)]=1$, $\alpha<1/2$ and if $(p_1,p_2)\in \ker(D)$ one obtains:
\begin{center}
\begin{tabular}{|c|c|c|c|c|c|c|}
\hline
 No & Conditions & $\#\sigma_p(H)$ & $\phi_1$ & $\phi_2$ & $\phi_3$  & $\sum_j\phi_j$
 \\ \hline\hline
IV.4 &$ \ell<0,p_1\neq 0 $&$ 1 $&$ 0 $&$ -2\pi\alpha $&$ -2\pi(1-\alpha) $&$ -2\pi$ \\\hline
IV.5 &$ \ell<0,p_1=0 $&$ 1 $&$ 0 $&$ -2\pi(1-\alpha) $&$ -2\pi\alpha $&$ -2\pi$ \\\hline
IV.6 &$ \ell>0,p_1\neq 0 $&$ 0 $&$ 0 $&$ 2\pi(1-\alpha) $&$ -2\pi(1-\alpha)$&$ 0$ \\\hline
IV.7 &$ \ell>0,p_1= 0 $&$ 0 $&$ 0 $&$ 2\pi\alpha $&$ -2\pi\alpha $&$ 0$ \\\hline
IV.8 &$ \ell=0,p_1\;\!p_2\neq 0 $&$ 0 $&$ 2\pi\alpha $&$ 2\pi(1-2\alpha) $&$ -2\pi(1-\alpha) $&$ 0$ \\\hline
IV.9.a &$ \ell=0,p_1= 0 $&$ 0 $&$ 2\pi\alpha $&$ 0 $&$ -2\pi\alpha$&$ 0$ \\\hline
IV.9.b &$ \ell=0,p_2= 0 $&$ 0 $&$ 2\pi(1-\alpha) $&$ 0 $&$ -2\pi(1-\alpha)$&$ 0$ \\\hline
\end{tabular}
\end{center}

Finally, if $\dim[\ker(D)]=1$, $\alpha>1/2$ and $(p_1,p_2)\in \ker(D)$ one has:
\begin{center}
\begin{tabular}{|c|c|c|c|c|c|c|}
\hline
 No & Conditions & $\#\sigma_p(H)$ & $\phi_1$ & $\phi_2$ & $\phi_3$  & $\sum_j\phi_j$
 \\ \hline\hline
IV.10 &$ \ell<0,p_2\neq 0 $&$ 1 $&$ 0 $&$ -2\pi(1-\alpha) $&$ -2\pi\alpha $&$ -2\pi$ \\\hline
IV.11 &$ \ell<0,p_2=0 $&$ 1 $&$ 0 $&$ -2\pi\alpha $&$ -2\pi(1-\alpha) $&$ -2\pi$ \\\hline
IV.12 &$ \ell>0,p_2\neq 0 $&$ 0 $&$ 0 $&$ 2\pi\alpha $&$ -2\pi\alpha$&$ 0$ \\\hline
IV.13 &$ \ell>0,p_2= 0 $&$ 0 $&$ 0 $&$ 2\pi(1-\alpha) $&$ -2\pi(1-\alpha) $&$ 0$ \\\hline
IV.14 &$ \ell=0,p_1\;\!p_2\neq 0 $&$ 0 $&$ 2\pi(1-\alpha) $&$ -2\pi(1-2\alpha) $&$ -2\pi\alpha $&$ 0$ \\\hline
IV.15.a &$ \ell=0,p_1= 0 $&$ 0 $&$ 2\pi\alpha $&$ 0 $&$ -2\pi\alpha$&$ 0$ \\\hline
IV.15.b &$ \ell=0,p_2= 0 $&$ 0 $&$ 2\pi(1-\alpha) $&$ 0 $&$ -2\pi(1-\alpha)$&$ 0$ \\\hline
\end{tabular}
\end{center}

\section{$\boldsymbol{K}$-groups, $\boldsymbol{n}$-traces and their pairings}\label{secK}

In this section, we give a very short account on the $K$-theory
for $C^*$-algebras and on various constructions related to it.
Our aim is not to present a thorough introduction to these subjects but
to recast the result obtained in the previous section in the most
suitable framework. For the first part, we refer to \cite{Roerdam} for an enjoyable
introduction to the subject.

\subsection{$K$-groups and boundary maps}

The $K_0$-group of a unital $C^*$-algebra $\E$ is constructed from
the homotopy classes
of projections in the set of square matrices with entries in $\E$.
Its addition is induced from the addition of two orthogonal projections:
if $p$ and $q$ are orthogonal projections, {\it i.e.}~$pq=0$, then also $p+q$ is a
projection. Thus, the sum of two homotopy classes $[p]_0+[q]_0$ is
defined as the class of the
sum of the block matrices $[p \oplus q]_0$ on the diagonal. This new class
does not depend on the choice of the representatives $p$ and $q$.
$K_0(\E)$ is defined as the Grothendieck group of
this set of homotopy classes of projections endowed with the mentioned addition.
In other words, the elements of the $K_0$-group are
given by formal differences:
$[p]_0-[q]_0$ is identified with $[p']_0-[q']_0$ if there exists a projection $r$
such that $[p]_0+[q']_0+[r]_0 = [p']_0+[q]_0+[r]_0$.
In the general non-unital case the
construction is a little bit more subtle.

The $K_1$-group of a $C^*$-algebra $\E$ is constructed from the homotopy classes
of unitaries in the set of square matrices with entries in the unitisation of $\E$.
Its addition is again defined by: $[u]_1+[v]_1 = [u\oplus v]_1$ as a block matrix on the
diagonal. The homotopy class of the added identity is the neutral element.

Now, let us consider three $C^*$-algebras $\J,\E$ and $\Q$ such that
$\J$ is an ideal of $\E$ and $\Q$ is isomorphic to the quotient
$\E/\J$. Another way of saying this is that $\J$ and $\Q$ are the
left and right part of an exact sequence of $C^*$-algebras
\begin{equation}\label{shortexact}
0\to \J\stackrel{\i}{\to} \E \stackrel{\q}{\to} \Q\to 0,
\end{equation}
$\i$ being an injective morphism and $\q$ a surjective morphism satisfying $\hbox{ker}\;\!
\q = \hbox{im}\;\! \i$.
There might not be any reasonable algebra morphism between $\J$
and $\Q$ but algebraic topology provides us with homomorphisms between their
$K$-groups: $\ind:K_1(\Q)\to K_0(\J)$ and $\exp:K_0(\Q)\to
K_1(\J)$, the index map and the exponential map. These maps are also referred to
as boundary maps.
For the sequel we shall be concerned only with the index map. It can be computed as follows: If $u$ is a unitary in $\Q$ then there exists a unitary
$w\in M_{2}(\E)$ such that
$\q(w)=\left( \begin{smallmatrix} u & 0 \\ 0 & u^*\end{smallmatrix}\right)$. It turns out that
$w\left(\begin{smallmatrix} 1 & 0 \\ 0 & 0 \end{smallmatrix}\right)w^*$ lies in the unitisation of $\i\big(M_2(\J)\big)$ so that
$\left(\left[w\left(\begin{smallmatrix} 1 & 0 \\ 0 & 0 \end{smallmatrix}
\right)w^*\right]_0
-\left[\left(\begin{smallmatrix} 1 & 0 \\ 0 & 0 \end{smallmatrix}
\right)\right]_0\right)$ defines an element of $K_0(\J)$. $\ind([u]_1)$ is that element.
 With a little luck there exists even a partial isometry $w\in \E$ such that $\q(w)=u$. Then  $(1-w^* w)$ and
$(1-w w^*)$ are projections in $\J$ and we have the simpler formula
\begin{equation}\label{eqformula}
\ind[u]_1 = \big[1-w^* w\big]_0-\big[1-w w^*\big]_0\ .
\end{equation}

\subsection{Cyclic cohomology, $n$-traces and Connes' pairing}\label{pourlesconstantes}

For this part, we refer to \cite[Sec.~III]{Connes} or to the short surveys presented in \cite[Sec.~5]{KRS} or in \cite[Sec.~4 \& 5]{KS}. For simplicity, we denote by $\N$ the set of natural number including $0$.

Given a complex algebra $\B$ and any $n \in \N$, let $C^n_\lambda(\B)$ be the set of $(n+1)$-linear functional on $\B$ which are cyclic in the sense that any $\eta \in C^n_\lambda(\B)$ satisfies for each $w_0,\dots,w_n\in \B$:
\begin{equation*}
\eta(w_1,\dots,w_n,w_0)=(-1)^n \eta(w_0,\dots,w_n)\ .
\end{equation*}
Then, let $\b: C^n_\lambda(\B) \to C^{n+1}_\lambda(\B)$ be the Hochschild coboundary map defined for $w_0,\dots,w_{n+1}\in \B$ by
\begin{equation*}
[\b \eta](w_0,\dots,w_{n+1}):=
\sum_{j=0}^n (-1)^j \eta(w_0,\dots,w_jw_{j+1},\dots ,w_{n+1}) +
(-1)^{n+1}\eta(w_{n+1}w_0,\dots,w_n)\ .
\end{equation*}
An element $\eta \in C^n_\lambda(\B)$ satisfying $\b\eta=0$ is called a cyclic $n$-cocyle, and the cyclic cohomology $HC(\B)$ of $\B$ is the cohomology of the complex
\begin{equation*}
0\to C^0_\lambda(\B)\to \dots \to C^n_\lambda(\B) \stackrel{\b}{\to}
C^{n+1}_\lambda(\B) \to \dots \ .
\end{equation*}

A convenient way of looking at cyclic $n$-cocycles is in terms of characters of a graded
differential algebra over $\B$. So, let us first recall that a graded differential
algebra $(\A,\d)$ is a graded algebra $\A$ together with a map $\d:\A\to \A$ of degree $+1$.
More precisely, $\A:=\oplus_{j=0}^\infty \A_j$ with each $\A_j$ an algebra over $\C$
satisfying the property $\A_j \;\!\A_k \subset \A_{j+k}$, and $\d$ is a graded
derivation satisfying $\d^2=0$.
In particular, the derivation satisfies
$\d(w_1w_2)=(\d w_1)w_2 + (-1)^{\deg(w_1)}w_1 (\d w_2)$, where $\deg(w_1)$
denotes the degree of the homogeneous element $w_1$.

A cycle $(\A,\d,\int)$ of dimension $n$ is a graded differential algebra
$(\A,\d)$, with $\A_j=0$ for $j>n$, endowed with a linear functional
$\int : \A \to \C$ satisfying $\int \d w=0$ if $w \in \A_{n-1}$ and for
$w_j\in \A_j$, $w_k \in \A_k$ :
$$
\int w_j w_k = (-1)^{jk}\int w_k w_j\ .
$$
Given an algebra $\B$, a cycle of dimension $n$ over $\B$ is a
cycle $(\A,\d,\int)$ of dimension $n$ together with a homomorphism $\rho:\B \to \A_0$.
In the sequel, we will assume that this map is injective and hence identify $\B$ with a subalgebra of $\A_0$ (and do not write $\rho$ anymore).
Now, if $w_0, \dots, w_n$ are $n+1$ elements of $\B$, one can define the character $\eta(w_0,\dots,w_n) \in \C$ by the formula:
\begin{equation}\label{defeta}
\eta(w_0,\dots,w_n):=\int w_0\;\!(\d w_1)\dots (\d w_n)\ .
\end{equation}
As shown in \cite[Prop.III.1.4]{Connes}, the map $\eta: \B^{n+1}\to \C$ is a cyclic $(n+1)$-linear functional on $\B$ satisfying
$\b\eta=0$, {\it i.e.} $\eta$ is a cyclic $n$-cocycle.
Conversely, any cyclic $n$-cocycle arises as the character of a cycle of dimension $n$ over $\B$. Let us also mention that a third description of any cyclic $n$-cocycle is presented in \cite[Sec.~III.1.$\alpha$]{Connes} in terms of the universal differential algebra associated with $\B$.

We can now introduce the precise definition of a $n$-trace over a Banach algebra. For an algebra $\B$ that is not necessarily unital, we denote by $\widetilde{\B}:= \B \oplus \C$ the algebra obtained by adding a unit to $\B$.
\begin{defin}\label{ntrace}
A $n$-trace on a Banach algebra $\B$ is the character of a cycle $(\A,\d,\int)$ of dimension $n$ over a dense subalgebra $\B'$ of $\B$ such that for all $w_1,\dots,w_n\in \B'$ and any $x_1,\dots,x_n \in \widetilde{\B}'$ there exists a constant $c= c(w_1,\dots,w_n)$ such that
$$
\left|\int (x_1 \d a_1)\dots (x_n \d w_n)\right| \leq c \|x_1\|\dots \|x_n\|\ .
$$
\end{defin}

\begin{rem}
Typically, the elements of $\B'$ are suitably smooth elements of $\B$
on which the derivation $\d$ is well defined and for which the r.h.s.~of \eqref{defeta}
is also well defined.
However, the action of the $n$-trace $\eta$ can sometines be extended to
more general elements $(w_0, \dots,w_n) \in \B^{n+1}$ by a suitable reinterpretation
of the l.h.s.~of \eqref{defeta}.
\end{rem}

The importance of $n$-traces relies on their duality relation with $K$-groups.
Recall first that $M_q(\B)\cong M_q(\C)\otimes \B$ and that $\tr$ denotes the standard
trace on matrices. Now, let $\B$ be a $C^*$-algebra and let $\eta_n$ be a $n$-trace
on $\B$ with $n \in \N$ even.
If $\B'$ is the dense subalgebra of $\B$ mentioned in Definition \ref{ntrace} and if
$p$ is a projection in $M_q(\B')$, then one sets
\begin{equation*}
\langle \eta_n,p\rangle := c_n\;\![\tr\otimes \eta_n](p,\dots,p) .
\end{equation*}
Similarly, if $\B$ is a unital $C^*$-algebra and if $\eta_n$ is a $n$-trace with $n \in \N$ odd,
then for any unitary $u$ in $M_q(\B')$  one sets
\begin{equation*}
\langle \eta_n,u\rangle := c_n \;\![\tr\otimes \eta_n](u^*,u,u^*,\dots,u)
\end{equation*}
the entries on the r.h.s.~alternating between $u$ and $u^*$. The constants $c_n$ are given by
$$
c_{2k}
\;=\;
\frac{1}{(2\pi i)^k}\,\frac{1}{k!}
\mbox{ , }
\qquad
c_{2k+1}
\;=\;
\frac{1}{(2\pi i)^{k+1}}\,
\frac{1}{2^{2k+1}}\,
\frac{1}{(k+\frac{1}{2})(k-\frac{1}{2})
\cdots\frac{1}{2}}
\mbox{ . }
$$

There relations are referred to as Connes' pairing between $K$-theory and cyclic cohomology
of $\B$ because of the following property, see \cite[Thm.~2.7]{Connes86} for a precise
statement and for its proof:
In the above framework, the values $\langle \eta_n,p\rangle$ and $\langle \eta_n,u\rangle$
depend only of the $K_0$-class $[p]_0$ of $p$ and of the $K_1$-class $[u]_1$ of $u$,
respectively.

We now illustrate these notions with two basic examples which will be of importance
in the sequel.

\begin{example}\label{exam1}
If $\B=\K(\HH)$, the algebra of compact operators on a Hilbert space $\HH$,
then the linear functional $\int$ on $\B$ is given by the usual trace $\Tr$
on the set $\K_1$ of trace class elements of $\K(\HH)$. Furthermore, since any projection
$p\in \K(\HH)$ is trace class, it follows that $\langle \eta_0,p\rangle\equiv
\langle \Tr,p\rangle$ is well defined
for any such $p$ and that this expression gives the dimension of the projection $p$.
\end{example}

For the next example, let us recall that $\det$ denotes the usual determinant of elements of $M_q(\C)$.

\begin{example}\label{exam2}
If $\B=C\big(\S^1, M_q(\C)\big)$ for some $q\geq 1$,
let us fix $\B':=C^1\big(\S^1,M_q(\C)\big)$. We parameterize $\S^1$ by the real numbers modulo $2\pi$ using $\theta$ as local coordinate.
As usual, for any $w \in \B'$
(which corresponds to an homogeneous element of degree $0$), one sets
$[\d w](\theta):=w'(\theta)\;\!\dd \theta$ (which is now an homogeneous element of degree $1$).
Furthermore, we define the graded trace $\int v \;\!\dd \theta :=\int_{-\pi}^{\pi} \tr[v(\theta)]\;\!\dd \theta$
 for an arbitrary element $v \;\!\dd \theta$ of degree $1$. This defines the $1$-trace $\eta_1$.
A  unitary element in $u\in C^1\big(\S^1,M_q(\C)\big)$ (or rather its class)  pairs as follows:
\begin{equation}\label{wn1}
\langle \eta_1,u \rangle =  c_1[\tr\otimes\eta_1](u^*,u) := \frac{1}{2\pi i}
\;\!\int_{-\pi}^\pi \tr[u(\theta)^*\;\!u'(\theta)]\;\! \dd \theta\ .
\end{equation}
For this example, the extension of this expression for any unitary $u \in
C\big(\S^1,M_q(\C)\big)$
is quite straightforward. Indeed, let us first rewrite $u=:e^{i\varphi}$ for some
$\varphi \in C^1\big(\S^1,M_q(\R)\big)$ and set $\beta(\theta):=\det[u(\theta)]$.
By using the equality $\det[e^{i\varphi}]=e^{i\tr[\varphi]}$, one then easily observed that
the quantity \eqref{wn1} is equal to
\begin{equation*}
\frac{1}{2\pi i}\int_{-\pi}^\pi \beta(\theta)^*\;\!\beta'(\theta)\;\!\dd \theta\ .
\end{equation*}
But this quantity is known to be equal to the winding number of the map
$\beta:\S^1 \to \T$, a quantity which is of topological nature and which only requires
that the map $\beta$ is continuous.
Altogether, one has thus obtained that the l.h.s.~of \eqref{wn1} is nothing but
the winding number of the map $\det[u]: \S^1 \to \T$, valid for any unitary $u \in
C\big(\S^1,M_q(\C)\big)$.
\end{example}

\subsection{Dual boundary maps}

We have seen that an $n$-trace $\eta$ over $\B$ gives rise to a functional on $K_i(\B)$ for $i=1$ or $i=2$, {\it i.e.}~the map $ \langle \eta,\cdot \rangle$ is an element of $Hom(K_i(\B),\C)$. In that sense $n$-traces are dual to the elements of the (complexified) $K$-groups. An important question is whether this dual relation is functorial in the sense that morphisms between the $K$-groups of different algebras yield dual morphisms on higher traces. Here we are in particular interested in a map on higher traces which is dual to the index map, {\it i.e.}~a map $\#$ which assigns to an even trace $\eta$ an odd trace $\#\eta$ such that
\begin{equation}\label{comp}
\langle \eta,\ind (\cdot) \rangle = \langle \#\eta,\cdot \rangle.
\end{equation}
This situation gives rise to
equalities between two numerical topological invariants.

Such an approach for relating two topological invariants has already been used at few occasions. For example, it has been recently shown that Levinson's theorem corresponds to a equality of the form \eqref{comp} for a $0$-trace and a $1$-trace \cite{KR1/2}.
In Section \ref{Sechigh} we shall develop such an equality for a $2$-trace and a $3$-trace.
On the other hand, let us mention that similar equalities have also been developed for the exponential map in \eqref{comp} instead of the index map. In this framework, an equality involving a $0$-trace and a $1$-trace has been put into evidence in \cite{Kel}. It gives rise to a relation
between the pressure on the boundary of a quantum system and the integrated density of states.
Similarly, a relation involving $2$-trace and a $1$-trace was involved in the proof of the equality between the bulk-Hall conductivity and the conductivity of the current along the edge of the sample, see \cite{KRS,KS}.

\section{Non-commutative topology and topological Levinson's theorems}\label{secAlgebra}

In this section we introduce the algebraic framework suitable for the Aharonov-Bohm model.
In fact, the following algebras were already introduced in
\cite{KRx} for the study of the wave operators in potential scattering on $\R$.
The similar form of the wave operators in the Aharonov-Bohm model and in the model studied
in that reference allows us to reuse part of the construction and the abstract results.
Let us stress that the following construction holds for fixed $\alpha$ and $(C,D)$.
These parameters will vary only at the end of the section.

\subsection{The algebraic framework}\label{lesalgebres}

For the construction of the $C^*$-algebras, let us introduce the operator
$B:=\frac{1}{2}\ln(H_0)$, where $H_0=-\Delta$ is the usual Laplace operator
on $\RR^2$. The crucial property of the operators $A$ and $B$ is that they satisfy the canonical commutation
relation $[A,B]=i$ so that $A$ generates translations in $B$ and vice versa,
\begin{equation*}
e^{iBt} A e^{-iBt} = A+t, \quad e^{iAs} B e^{-iAs} = B-s.
\end{equation*}
Furthermore, both operators leave the subspaces $\HH_m$ invariant.
More precisely, for any essentially bounded functions $\varphi$ and $\eta$ on $\R$,
the operator $\varphi(A)\eta(B)$ leaves each of these subspaces invariant.
Since all the interesting features of the Aharonov-Bohm model take place in the subspace
$\HH_\2\cong L^2(\R_+,r\;\!\dd r)\otimes \C^2$, we shall subsequently restrict our
attention to this subspace and consider functions $\varphi, \eta$ defined on $\R$
and taking values in $M_2(\C)$.

Now, let $\E$ be the closure in $\B(\HH_\2)$ of the algebra generated by elements of the form
$\varphi(A)\psi(H_0)$, where $\varphi$ is a continuous function on $\R$ with values in
$M_2(\C)$ which converges at $\pm \infty$, and $\psi$ is a continuous function $\R_+$
with values in $M_2(\C)$ which converges at $0$ and at $+\infty$.
Stated differently, $\varphi\in C\big(\bR,M_2(\C)\big)$ with $\bR=[-\infty,+\infty]$,
and $\psi \in C\big(\bRp,M_2(\C)\big)$ with $\bRp=[0,+\infty]$.
Let $\J$ be the norm closed algebra generated by $\varphi(A)\psi(H_0)$ with functions
$\varphi$ and $\psi$ for which the above limits vanish. Obviously, $\J$ is an ideal in $\E$,
and the same algebras are obtained if $\psi(H_0)$ is replaced by $\eta(B)$ with
$\eta \in C\big(\bR,M_2(\C)\big)$ or $\eta \in C_0\big(\R,M_2(\C)\big)$, respectively.
Furthermore, the ideal $\J$ is equal to the algebra of compact operators $\K(\HH_\2)$,
as shown in \cite[Sec.~4]{KRx}.

Let us already mention the reason of our interest in defining the above algebra $\E$.
Since for $m \in \{0,-1\}$ the functions $\varphi^-_m$ and $\tilde{\varphi}_{m}$ have limits
at $\pm \infty$, and since $\widetilde{S}^{\CD}_\alpha$ also has limits at $0$ and $+\infty$,
it follows from \eqref{yoyo} that the operator $\Omega_-^{\CD}|_{\HH_\2}$ belongs to $\E$.
Since $\J = \K(\HH_\2)$, the image $\q\big(\Omega_-^{\CD}|_{\HH_\2}\big)$ in $\E/\J$
corresponds to the image of $\Omega_-^{\CD}|_{\HH_\2}$ in the Calkin algebra.
This motivates the following computation of the quotient $\E/\J$.

To describe the quotient $\E/\J$ we consider the square
$\blacksquare:=\bRp\times \bR$ whose boundary $\square$ is the
union of four parts: $\square =B_1\cup B_2\cup B_3\cup
B_4$, with $B_1 = \{0\}\times \bR$, $B_2 = \bRp \times \{+\infty\}$,
$B_3 = \{+\infty\}\times \bR$ and $B_4 = \bRp\times \{-\infty\}$. We
can then view $\Q:=C\big(\square,M_2(\C)\big)$ as the subalgebra
of
\begin{equation*}
C\big(\bR,M_2(\C)\big)\oplus C\big(\bRp,M_2(\C)\big)
\oplus C\big(\bR,M_2(\C)\big)\oplus C\big(\bRp,M_2(\C)\big)
\end{equation*}
given by elements
$(\Gamma_1,\Gamma_2,\Gamma_3,\Gamma_4)$ which coincide at the
corresponding end points, that is,
$\Gamma_1(+\infty)=\Gamma_2(0)$,
$\Gamma_2(+\infty)=\Gamma_3(+\infty)$,
$\Gamma_3(-\infty)=\Gamma_4(+\infty)$,
$\Gamma_4(0)=\Gamma_1(-\infty)$.
The following lemma corresponds to results obtained in
\cite[Sec.~3.5]{Georgescu} rewritten in our framework.

\begin{lemma}\label{image}
$\E/\J$ is isomorphic to $\Q$.
Furthermore, for any $\varphi\in C\big(\bR,M_2(\C)\big)$ and for any
$\psi \in C\big(\bRp,M_2(\C)\big)$, the image of $\varphi(A)\psi(H_0)$
through the quotient map $\q: \E \to \Q$ is given by
$\Gamma_1(\cdot) = \varphi(\cdot)\psi(0)$, $\Gamma_{2}(\cdot) =
\varphi(+\infty)\psi(\cdot)$,
$\Gamma_{3}(\cdot) = \varphi(\cdot)\psi(+\infty)$ and $\Gamma_{4}(\cdot) =
\varphi(-\infty)\psi(\cdot)$.
\end{lemma}

Stated differently, the algebras $\J, \E$ and $\Q$ are part of the short exact sequence of $C^*$-algebras \eqref{shortexact}.
And as already mentioned, the operator $\Omega_-^{\CD}|_{\HH_\2}$ clearly belongs to $\E$. Furthermore, its image through the quotient map $\q$ can be easily computed, and in fact has already been computed. Indeed, the function $\Gamma(C,D,\alpha,\cdot)$ presented in \eqref{fctGamma} is precisely $\q\big(\Omega_-^{\CD}|_{\HH_\2}\big)$, as we shall see it in the following section.

\begin{rem}\label{remautre}
We still would like to provide an alternative description of the above algebras and of the corresponding short exact sequence.
Recall that
$\Q$ is isomorphic to $C\big(\T,M_2(\C)\big)$ and that $C(\T)$ is as a
$C^*$-algebra generated by a single
continuous bijective function $u:\T\to \T\subset\C$ with winding number $1$.
There are, up to a natural equivalence, not so many $C^*$-algebras
$\B$ which fit into an exact sequence of the form
$0\to \K\stackrel{}\to\B\stackrel{}\to C(\T)\to 0$, with $\K$ the algebra of compact operators.
In fact, it turns out that they are classified by the Fredholm-index of a
lift $\hat u$  of $u$ \cite[Thm.~IX.3.3]{Davidson}.
In the present case we can use an exactly solvable model to find out
that $\hat u$ can be taken to be an isometry of co-rank $1$ and hence
this index is $-1$ \cite{KRx}. Our extension is thus what one refers
to as the Toeplitz extension.
This means that $\E$ is the tensor product of $M_2(\C)$ with
the $C^*$-algebra generated by an element $\hat u$
satisfying  $\hat u^*\hat u = 1$ and $\hat u\hat u^* = 1 - e_{00}$
where $e_{00}$ is a rank $1$ projection. The surjection $\q$ is
uniquely defined by $\q(\hat u) = u$. Our exact sequence is thus the tensor
product with $M_2(\C)$ of the exact sequence
\begin{equation}\label{eqautre}
0 \to \K \stackrel{\i}\to C^*(\hat u) \stackrel{\q }{\to}C^*(u) \to 0.
\end{equation}
\end{rem}

\subsection{The $0$-degree Levinson's theorem, the topological approach}

We can now state the topological version of Levinson's theorem.

\begin{theorem}\label{Ktheo}
For each $\alpha \in (0,1)$ and each admissible pair $(C,D)$, one has
$\Omega_-^{\CD}|_{\HH_\2}\in \E$. Furthermore,
$\q\big(\Omega_-^{\CD}|_{\HH_\2}\big)=\Gamma(C,D,\alpha,\cdot) \in \Q$
and the following equality holds
\begin{equation*}
\ind[\Gamma(C,D,\alpha,\cdot)]_1 = -[P^{\CD}_\alpha]_0\ ,
\end{equation*}
where $P^{\CD}_\alpha$ is the orthogonal projection on the space spanned by the bound states of $H^{\CD}_\alpha$.
\end{theorem}

\begin{rem}
Recall that by Atkinson's theorem the image of
any Fredholm operator $F\in\B(\HH_\2)$ in the Calkin algebra $\B(\HH_\2)/ \K(\HH_\2)$
is invertible. Then, since the wave operators $\Omega_-^{\CD}|_{\HH_\2}$ is an isometry
and a Fredholm operator,
it follows that each function $\Gamma_j(C,D,\alpha,\cdot)$ takes values in $U(2)$.
In fact, this was already mentioned when the functions $\Gamma_j(C,D,\alpha,\cdot)$ were introduced.
\end{rem}

\begin{proof}[Proof of Theorem~\ref{Ktheo}]
The image of $\Omega_-^{\CD}|_{\HH_\2}$ through the quotient map $\q$ is easily obtained by taking
the formulae recalled in Lemma \ref{image} into account.
Then, since $\Omega_-^{\CD}|_{\HH_\2}$ is a lift for $\Gamma(C,D,\alpha,\cdot)$,
the image of $[\Gamma(C,D,\alpha,\cdot)]_1$ though the index map is obtained by the formula \eqref{eqformula}:
\begin{eqnarray*}
\ind[\Gamma(C,D,\alpha,\cdot)]_1 &=& \big[1-\big(\Omega_-^{\CD}|_{\HH_\2}\big)^* \;\!\Omega_-^{\CD}|_{\HH_\2}\big]_0-\big[1-\Omega_-^{\CD}|_{\HH_\2}\;\! \big(\Omega_-^{\CD}|_{\HH_\2}\big)^*\big]_0\\
&=& [0]_0-\big[P^{\CD}_\alpha\big]_0\ .
\end{eqnarray*}
\end{proof}

Theorem~\ref{Ktheo} covers the $K$-theoretic part of Levinson's theorem. In order to get a genuine Levinson's theorem, by which we mean an equality between topological numbers, we need to add the dual description, {\it i.e.}~identify higher traces on $\J$ and $\Q$ and a dual boundary map. As a matter of fact, the algebras considered so far are too simple to allow for non-trivial results in higher degree
and so we must content ourselves here to identify a suitable $0$-trace and $1$-trace which can be
applied to $P^{\CD}_\alpha$ and $\Gamma(C,D,\alpha,\cdot)$, respectively.
Clearly, only the usual trace $\Tr$ can be applied on the former term, {\it cf.}~Example \ref{exam1}
of Section \ref{secK}.
On the other hand, since $\Gamma(C,D,\alpha,\cdot) \in C\big(\square,U(2)\big)$,
we can define the winding number $\wind\big[\Gamma(C,D,\alpha,\cdot)\big]$ of the map
\begin{equation*}
\square \ni \zeta \mapsto \det [\Gamma(C,D,\alpha,\zeta)]\in \T
\end{equation*}
with orientation of $\square$ chosen clockwise, {\it cf.}~Example \ref{exam2}
of Section \ref{secK}.
Then, the already stated Theorem \ref{Lev0} essentially reformulates the fact that the $0$-trace is mapped to the $1$-trace
by the dual of the index map. The first equality of  Theorem \ref{Lev0} can then be found in Proposition 7 of \cite{KRx} and the equality between the cardinality of $\sigma_p(H_\alpha^{\CD})$ and the number of
negative eigenvalues of the matrix $CD^*$ has been shown in \cite[Lem.~4]{PR}.

\subsection{Higher degree Levinson's theorem}\label{Sechigh}

The previous theorem is a pointwise $0$-degree Levinson's theorem. More precisely, it was obtained for fixed $C,D$ and $\alpha$. However, it clearly calls for making these parameters
degrees of freedom and thus to include them into the description of the algebras. In the context of our physical model this amounts to considering
families of self-adjoint extensions of $H_\alpha$.
For that purpose we use the one-to-one parametrization of these extensions with elements $U \in U(2)$
introduced in Remark \ref{1to1}. We denote the self-adjoint extension corresponding to $U \in U(2)$ by $H_\alpha^U$.

So, let us consider a smooth and compact orientable $n$-dimensional manifold $X$ without boundary.
Subsequently, we will choose for $X$ a two-dimensional submanifold of $U(2)\times (0,1)$.
Taking continuous functions over $X$ we get a new short exact sequence
\begin{equation}\label{degresup}
0 \to C(X,\J)\stackrel{}\to C(X,\E) \stackrel{}\to C(X,\Q)\to 0\ .
\end{equation}
Furthermore, recall that $\J$ is endowed with a $0$-trace and the algebra $\Q$ with
a $1$-trace.
There is a standard construction in cyclic cohomology, the cup
product, which provides us with
a suitable $n$-trace on the algebra $C(X,\J)$ and a corresponding $n+1$-trace
on the algebra $C(X,\Q)$, see \cite[Sec.~III.1.$\alpha$]{Connes}. We describe it here in terms of cycles.

Recall that any smooth and compact manifold $Y$ of dimension $d$
naturally defines a structure of a graded differential algebra $(\A_\Y,\d_\Y)$,
the algebra of its smooth differential $k$-forms.
If we assume in addition that $Y$ is orientable so that we can choose
a global volume form,
then the linear form $\int_\Y$ can be defined by integrating the
$d$-forms over $Y$.
In that case, the algebra $C(Y)$ is naturally endowed with the $d$-trace defined
by the character of the cycle $(\A_\Y,\d_\Y,\int_\Y)$ of dimension $d$
over the dense subalgebra $C^\infty(Y)$.

For the algebra $C(X,\J)$, let us recall that $\J$ is equal to the
algebra $\K(\HH_\2)$
and that the $0$-trace on $\J$ was simply the usual trace $\Tr$.
So, let $\K_1$ denote the trace class elements of $\K(\HH_\2)$.
Then, the natural graded differential algebra associated with
$C^\infty(X,\K_1)$ is given by $(\A_\X\otimes \K_1,\d_\X)$.
The resulting $n$-trace on $C(X,\J)$ is then defined by the character
of the cycle $(\A_\X\otimes \K_1,\d_\X,\int_\X\otimes \Tr)$ over the dense
subalgebra $C^\infty(X,\K_1)$ of $C(X,\J)$. We denote it by $\eta_X$.

For the algebra $C(X,\Q)$, let us recall that $\Q =C\big(
\square,M_2(\C)\big)$
with $ \square \cong \S^1$, and thus
$C(X,\Q) \cong C\big(X\times\S^1,M_2(\C)\big) \cong C(X \times \S^1)
\otimes M_2(\C)$.
Since $X\times \S^1$ is a compact orientable manifold without
boundary, the above construction
applies also to $C\big(X\times\S^1,M_2(\C)\big)$. More precisely,
the exterior derivation on $X\times \S^1$ is the sum of  $\d_\X$ and
$\d_{\S^1}$ (the latter was denoted simply by $\d$ in Example
\ref{exam2}). Furthermore, we consider the natural volume form on $X\times \S^1$.
Note because of the factor $M_2(\C)$ the graded trace of the cycle
involves the usual matrix  trace $\tr$.
Thus the resulting $n+1$-trace is the character of the cycle
$(\A_{\X\times \S^1}\otimes M_2(\C),\d_\X+\d_{\S^1},\int_{\X\times
  \S^1}\otimes \tr)$. We denote it by $\#\eta_X$.

Having these constructions at our disposal we can now state the main result of this section.
For the statement, we use the one-to-one parametrization of the extensions of $H_\alpha$
introduced in Remark \ref{1to1} and let $\alpha\in(0,1)$.
We  consider a family $\{\Omega_-(H_\alpha^U,H_0)\}_{(U,\alpha)\in X} \in \B(\HH_\2)$,
parameterized by some compact orientable and boundaryless submanifold $X$ of $U(2)\times (0,1)$. This family defines a map $\bOmega:X\to \E$,
$\bOmega(U,\alpha) = \Omega_-(H_\alpha^U,H_0)$, a map $\bGamma:X\to \Q$,
$\bGamma(U,\alpha,\cdot) = \Gamma\big(C(U),D(U),\alpha,\cdot\big)= \q(\Omega_-(H_\alpha^U,H_0))$, and
a map $\bP:X\to \J$,
$\bP(U,\alpha) = P_\alpha^U$  the orthogonal
projection of the subspace of $\HH_\2$ spanned by the bound states of $H_\alpha^U$.

\begin{theorem}\label{thm-ENN}
Let $X$ be a smooth, compact and orientable $n$-dimensional submanifold of $U(2)\times (0,1)$ without boundary.
Let us assume that the map $\bOmega:X\to \E$
is continuous. Then the following equality holds:
\begin{equation*}
\ind[\bGamma]_1 = -[\bP]_0
\end{equation*}
where $\ind$ is the index map from $K_1\big(C(X,\Q)\big)$ to $K_0\big(C(X,\J)\big)$.
Furthermore, the numerical equality
\begin{equation}\label{eq23}
\big\langle \#\eta_{X},[\bGamma]_1 \big\rangle
=-
\big \langle \eta_{X},[\bP]_0 \big\rangle
\end{equation}
also holds.
\end{theorem}

\begin{proof}
For the first equality we can simply repeat pointwise the proof of Theorem~\ref{Ktheo}. Since we required $\bOmega$ to be continuous, its kernel projection $\bP$ is continuous as well.
The second equality follows from of a more general formula stating that the map $\eta_X\mapsto \#\eta_X$ is dual to the boundary maps \cite{ENN88}.  We also mention that another proof can be obtained by mimicking the calculation given in the Appendix of \cite{KRS}. For the convenience of the reader, we
sketch it in the Appendix~\ref{appa} and refer to \cite{KRS} for details.
\end{proof}

Let us point out that r.h.s.~of \eqref{eq23} is the Chern number of the vector bundle given by the eigenstates of $H_\alpha^U$. The next section is devoted to a computation of this number for a special choice of manifold $X$.

\subsection{An example of a non trivial Chern number}

We shall now choose a $2$-dimensional manifold $X$ and show that the above relation between the corresponding $2$-trace and $3$-trace is not trivial. More precisely, we shall choose a manifold $X$ such that the r.h.s. of \eqref{eq23} is not equal to $0$.

For that purpose, let us fix two complex numbers  $\lambda_1,\lambda_2$ of modulus $1$
with $\Im \lambda_1<0< \Im \lambda_2$ and consider the set $X \subset U(2)$ defined by :
\begin{equation*}
X = \left\{ V
\left(\begin{smallmatrix}
\lambda_1 & 0 \\
0 & \lambda_2
\end{smallmatrix}\right)
 V^* \mid V\in U(2)\right\}.
\end{equation*}
Clearly, $X$ is a two-dimensional smooth and compact manifold without boundary, which can be parameterized by
\begin{equation}\label{param}
X= \left\{
\left(\begin{matrix}
\rho^2 \lambda_1 + (1-\rho^2) \lambda_2 &
 \rho(1-\rho^2)^{1/2} \;\!e^{i\phi}(\lambda_1-\lambda_2) \\
\rho(1-\rho^2)^{1/2} \;\!e^{-i\phi}(\lambda_1-\lambda_2) &
(1-\rho^2) \lambda_1 + \rho^2 \lambda_2
\end{matrix}\right)
\mid \rho \in [0,1] \hbox{ and } \phi \in [0,2\pi)\right\}.
\end{equation}
Note that the $(\theta,\phi)$-parametrization of $X$ is complete in the sense that it covers all the manifold injectively away from a subset of codimension $1$,
but it has coordinate singularities at $\rho\in \{0,1\}$.

By \cite[Lem.~15]{PR}, for each $U\equiv U(\rho,\phi)\in X$ the operator $H^U_\alpha$ has a single negative eigenvalue $z\equiv z(U)$
defined by the equality $\det \big( (1+U) M(z) +i (1-U)\big)=0$, and one has
\begin{equation}\label{eq-khu}
\ker (H^U_\alpha -z)=\gamma(z) \ker \big( (1+U) M(z) +i (1-U)\big).
\end{equation}
Here, $M(z)$ is the Weyl function which is a $2\times 2$ diagonal matrix
and $\gamma(z):\C^2\to \HH$ an injective linear map (see subsection \ref{ssec21}).
The orthogonal projection onto $\ker (H^U_\alpha -z)$ is denoted by $P^U_\alpha$ and we shall consider $E=\{\im P^U_\alpha\mid U \in X\}$ which is a subbundle of the trivial bundle $X\times \HH$.
Our next aim is to calculate its Chern number $\ch(E)$, first in terms of the Chern number of a simpler bundle. In view of \eqref{eq-khu}  $X\times\C^2 \ni (U,\xi) \mapsto (U,\gamma(z(U))\xi)\in X\times\HH$ defines a continuous isomorphism between the subbundle $F$ of the trivial bundle $X\times \CC^2$ defined by
\begin{equation*}
F=\big\{\ker \big( (1+U) M(z) +i (1-U)\big)\mid U \in X\big\}.
\end{equation*}
and $E$, and hence
 $\ch(E)=\ch(F)$.
Now, the assumptions on $\lambda_1$ and $\lambda_2$ imply that for any $U \in X$ the matrix $(1-U)$ is invertible and one can then consider the self-adjoint operator
\[
T(U)=i\,\dfrac{1-U}{1+U}\ .
\]
By setting $\lambda_j =: e^{i\varphi_j}$ with $\varphi_1 \in (-\pi,0)$ and $\varphi_2 \in (0,\pi)$, and then $r_i = \tan \frac{\varphi_i}{2}$ we get
\[ T(U)
= \left(\begin{matrix}
\rho^2r_1 +(1-\rho^2)r_2 & \rho (1-\rho^2)^{1/2}e^{i\phi}(r_1-r_2) \\ \rho (1-\rho^2)^{1/2}e^{-i\phi}(r_1-r_2)  &
(1-\rho^2)r_1 +\rho^2 r_2
\end{matrix}\right)
\]
for some $\rho \in [0,1]$ and $\phi \in [0,2\pi)$ given by \eqref{param}.
Thus, by using the parametrization of $U$ and $z$ in terms of $(\rho,\phi)$ one obtains that the bundle $E$ is
isomorphic to the bundle $G$ defined by
\begin{equation*}
G=\big\{\ker \big( G(\rho,\phi)\big)
\mid \rho \in [0,1] \hbox{ and } \phi \in [0,2\pi)\big\}.
\end{equation*}
with
\[G(\rho,\phi)
:= \left(\begin{matrix}
M_{11}\big(z(\rho,\phi)\big)+\rho^2r_1 +(1-\rho^2)r_2 & \rho (1-\rho^2)^{1/2}e^{i\phi}(r_1-r_2) \\ \rho (1-\rho^2)^{1/2}e^{-i\phi}(r_1-r_2)  &
M_{22}\left(z(\rho,\phi)\right)+ (1-\rho^2)r_1 +\rho^2 r_2
\end{matrix}\right)\ .
\]
Recall that $z(\rho,\phi)$ is defined by the condition $\det \big(G(\rho,\phi)\big)=0$, {\it i.e.}
\begin{equation}\label{eq-annul}
\big(
M_{11}\big(z(\rho,\phi)\big)+\rho^2r_1 +(1-\rho^2)r_2
\big)\cdot
\big(
M_{22}\left(z(\rho,\phi)\right)+ (1-\rho^2)r_1 +\rho^2 r_2
\big)=(r_1-r_2)^2 (1-\rho^2)\rho^2.
\end{equation}
Finally, since $M(z)$ is self-adjoint for $z \in \R_-$, the matrix  $G(\rho,\phi)$ is  self-adjoint, and hence
$\ker G(\rho,\phi) = \big(\im G(\rho,\phi)\big)^\perp$.
In particular, if one defines the bundle
\begin{equation}\label{eq-bf}
H=\big\{\im G(\rho,\phi)\mid \rho \in [0,1] \hbox{ and } \phi \in [0,2\pi)\big\}
\end{equation}
one obviously has $G+H=X\times \CC^2$, and then
$\ch(G)=-\ch(H)$ as the Chern number of the trivial bundle $X\times \CC^2$ is zero.
In summary, $\ch(E) = -\ch(H)$, which we are going to calculate after the following remark.

\begin{rem}\label{exchern}
Let $A:X\to M_2(\CC)$ be a continuously differentiable map with $A(x)$ of rank $1$ for all $x \in X$. Let us recall how to
calculate the Chern number of the bundle $B=\{\im A(x)\mid x\in X\}$.
Assume that the first column $A_1$ of $A$ vanishes only on a finite set $Y$.
If $Y$ is empty, the bundle is trivial and $\ch (B)=0$.
So let us assume that $Y$ is non-empty.
Let $P(x)$ be the matrix of the orthogonal projection onto $\im A(x)$ in $\CC^2$.
By definition, one has
\[
\ch(B)=\dfrac{1}{2\pi i}\int_X \tr \big(P \;\d_X P\wedge \d_X P\big).
\]
Now, for $\epsilon>0$ consider an open set $V_\epsilon\subset X$ with $Y\subset V_\epsilon$, having a $C^1$ boundary
and such that $\vol_X V_\epsilon\to 0$ as $\epsilon\to 0$.
By continuity and compactness, the differential form $\tr \big(P \; \d_X P \wedge\d_X P\big)$ is bounded, and then
\[
\ch(B)=\dfrac{1}{2\pi i}\lim_{\epsilon\to 0}\int_{X \setminus V_\epsilon} \tr \big(P \;\d_X P\wedge \d_X P\big).
\]
For $x\in X \setminus V_\epsilon$ one can consider the vector
\[
\psi(x)=\dfrac{A_1(x)}{\|A_1(x)\|}
\]
and by a direct calculation, one obtains $\tr \big(P \; \d_X P\wedge \d_X P\big)= \d_X\Bar \psi \wedge \d_X\psi$.
Since the differential form $\d_X\Bar \psi \wedge \d_X\psi$ is exact, then $\d_X \Bar \psi \wedge \d_X\psi=\d_X(\Bar \psi \;\d_X\psi)$ and by Stokes' theorem,
one obtains
\[
\ch(B)=
\dfrac{1}{2\pi i}\lim_{\epsilon\to 0}
\int_{\partial V_\epsilon} \Bar \psi\;\d_X\psi.
\]
\end{rem}

Let us apply the above constructions to the bundle \eqref{eq-bf}.
Since $(r_1 -r_2)\neq 0$ the first column $G_1(\rho,\phi)$ of the matrix $G(\rho,\phi)$
can potentially vanish only for $\rho=0$ or for $\rho=1$.
As already mentioned, these two points are the coordinate singularities of the parametrization. But by a local change of parametrization, one easily get rid of this pathology.
Thus, we first consider $\rho=1$ and let $(\theta_1, \theta_2) \in (-1,1)^2$ be a local parametrization of a neighbourhood of the point $\rho=1$ which coincides with $(\theta_1, \theta_2)=(0,0)$.
Let $\widetilde{G}$ be the expression of the function $G$ in the coordinates $(\theta_1, \theta_2)$ and in a neighbourhood of the point $\rho=1$. For this function one has
\[
\widetilde{G}_1(0,0) = \begin{pmatrix}
M_{11} \big(z(0,0)\big) + r_1\\
0
\end{pmatrix}
\]
Now, note that under our assumptions one has $r_1<0$ and $r_2>0$.
As seen from the explicit expressions for $M$,
the entries of $M(z)$ are negative for $z<0$.
Then the term $M_{11} \big(z(0,0)\big) + r_1$ can not be equal to
$0$ and this also holds for the first coefficient of $\widetilde{G}_1(0,0)$.

For $\rho =0$ let $(\vartheta_1, \vartheta_2) \in (-1,1)^2$ be a local parametrization of a neighbourhood of the point $\rho=0$ which coincides with $(\vartheta_1, \vartheta_2)=(0,0)$.
Let again $\widehat{G}$ be the expression of the function $G$ in the coordinates $(\vartheta_1, \vartheta_2)$ and in a neighbourhood of the point $\rho=0$. Then one has
\[
\widehat{G}_1(0,0) = \begin{pmatrix}
M_{11} \big(z(0,0)\big) + r_2\\
0
\end{pmatrix}\ .
\]
In that case, since $M_{22}(z)+ r_1$ is strictly negative for any $z\in \R_-$ one must have $M_{11} \big(z(0,0)\big) + r_2=0$ in order to satisfy Equation \eqref{eq-annul}.
Therefore, the corresponding point $\rho=0$ belongs to $Y$, as introduced in Remark \ref{exchern}.
Therefore, in our case $Y$ consists in a single point $y$ corresponding to $\rho=0$.

Now, for $\epsilon>0$ consider the set
\begin{equation*}
V_\epsilon= \left\{
\left(\begin{matrix}
\rho^2 \lambda_1 + (1-\rho^2) \lambda_2 &
 \rho(1-\rho^2)^{1/2} \;\!e^{i\phi}(\lambda_1-\lambda_2) \\
\rho(1-\rho^2)^{1/2} \;\!e^{-i\phi}(\lambda_1-\lambda_2) &
(1-\rho^2) \lambda_1 + \rho^2 \lambda_2
\end{matrix}\right)
\mid \rho \in [0,\epsilon) \hbox{ and } \phi \in [0,2\pi)\right\}.
\end{equation*}
Obviously, this set satisfies the conditions of Remark \ref{exchern}. We can then represent
\[
G_1(\rho,\phi) = \begin{pmatrix}
M_{11}\big(z(\rho,\phi)\big)+\rho^2r_1 +(1-\rho^2)r_2 \\
\rho (1-\rho^2)^{1/2}e^{-i\phi}(r_1-r_2)
\end{pmatrix}
=:
\begin{pmatrix}
g(\rho,\phi)\\
f(\rho) e^{-i\phi}
\end{pmatrix}
\]
with $f,g$ real, and set
\[
\psi(\rho,\phi):=\dfrac{G_1(\rho,\phi)}{\|G_1(\rho,\phi)\|}=
\begin{pmatrix}
\dfrac{g(\rho,\phi)}{\sqrt{f^2(\rho)+g^2(\rho,\phi)}}\\
\dfrac{f(\rho) e^{-i\phi}}{\sqrt{f^2(\rho)+g^2(\rho,\phi)}}
\end{pmatrix}.
\]
Then one has
\begin{eqnarray*}
\int_{\partial V_\epsilon(y)} \Bar\psi\,\d_X \psi
&=&
\int_0^{2\pi} \Big[ \dfrac{g}{\sqrt{f^2+g^2}}\;\!\partial_\phi \Big(\dfrac{g}{\sqrt{f^2+g^2}}\Big)
+ \dfrac{f e^{i\cdot}}{\sqrt{f^2+g^2}}\;\!
\partial_\phi\Big(\dfrac{f e^{-i\cdot}}{\sqrt{f^2+g^2}}\Big)
\Big](\epsilon,\phi) \;\!\dd\phi \\
&=&
-i\int_0^{2\pi} \Big[\dfrac{f^2}{f^2+g^2}\Big](\epsilon,\phi) \dd\phi
\end{eqnarray*}

Thus, one has obtained that
\begin{equation}\label{eq-cf}
\ch(H)=-\dfrac{1}{2\pi}\lim_{\epsilon\to 0} \int_0^{2\pi} \dfrac{f^2(\epsilon)}{f^2(\epsilon)+g^2(\epsilon,\phi)} \,\dd\phi.
\end{equation}
Furthermore, note that Equation \eqref{eq-annul}
can be rewritten as
$g(\rho,\phi) h(\rho,\phi)=f^2(\rho)$,
where $h(\rho,\phi)=\big(M_{22}\left(z(\rho,\phi)\right)+ (1-\rho^2)r_1 +\rho^2 r_2\big)$ does not vanish in a sufficiently small neighbourhood of the point $\rho=0$.
Then one  has  $g(\rho,\phi)=o\big(f(\rho)\big)$ uniformly in $\phi$ as $r$ tends to $0$. By
substituting this observation into \eqref{eq-cf}
one obtains
\[
\ch(H)=-\dfrac{1}{2\pi} \int_0^{2\pi} \lim_{\epsilon\to 0}\dfrac{f^2(\epsilon)}{f^2(\epsilon)+g^2(\epsilon,\phi)} \,\dd\phi
=-\dfrac{1}{2\pi} \int_0^{2\pi}  \,\dd\phi=-1.
\]
As a consequence, by returning to the original bundle $E$, one has obtained  $\ch(E)=-\ch(H)=1$.

As a corollary, one easily prove:

\begin{prop}\label{propfinal}
Let $\lambda_1,\lambda_2$ be two complex numbers of modulus $1$
with $\Im \lambda_1<0< \Im \lambda_2$ and consider the set $X \subset U(2)$ defined by \eqref{param}.
Then the map $\bOmega: X \to \E$
is continuous and the following equality holds:
\begin{equation*}
\frac{1}{24\pi^2} \int_{X\times \square}\tr\big[\bGamma^* \; \d_{X\times\square}\bGamma\wedge
\d_{X\times\square}\bGamma^* \wedge\d_{X\times\square}\bGamma \big]
= 1
\end{equation*}
\end{prop}

\begin{proof}
Continuity of $X \ni U \mapsto \Omega_-(H_\alpha^U,H_0) \in \E$ is proved in Appendix~\ref{appc}.
The equation is an application of Theorem~\ref{thm-ENN} with $n=2$ with $\eta_X$ defined by the first Chern character over $X$:
$\langle\eta_X, [\bP]_0 \rangle = \frac{1}{2\pi i} \int_{X}\Tr\big[ \bP \;\d_X\bP \wedge\d_X\bP\big] = \ch(E)$.
\end{proof}

\appendix

\newcommand{\jj}{\mathfrak j}

\section{Proof of Theorem \ref{thm-ENN}}\label{appa}

As already mentioned, we simply sketch the proof of the second equality of Theorem \ref{thm-ENN} mimicking the approach of the Appendix of \cite{KRS}.
Note that this proof is based on the alternative description of the $C^*$-algebras provided in Remark \ref{remautre}.

\begin{proof}[Proof of the second statement of Theorem \ref{thm-ENN}]
1) Let us first observe that the short exact sequence \eqref{eqautre} illuminates better the $K$-theory associated with the relevant algebras.
Indeed, the relations for $\hat u$ tell us immediately that
$1-e_{00}$ and $1$ are Murray-von Neumann equivalent and hence define the same $K_0$-element in $\E$.
It follows that the two maps $K_0(\i):K_0(\K)\to K_0\big(C^*(\hat u)\big)$ and $K_1(\i):K_1(\K)\to K_1\big(C^*(\hat u)\big)$ are the zero maps, so that the six-term exact
sequence splits into two short exact sequences, see \cite[Chap.~12]{Roerdam} for more information on the six-term exact
sequence.
>From this one may conclude that the inclusion $\j:\C\ni 1\mapsto 1\in C^*(\hat u)$ induces an isomorphism in $K$-theory.
The two exact sequences in $K$-theory therefore become for $i=0,1$ mod $2$:
$$ 0\to K_i(\C) \stackrel{K_i(\j)}{\longrightarrow} K_i(C^*(u))
\stackrel{\delta_i}\to K_{i-1}(\K)\to 0,$$
where $\delta_i$ are the boundary maps, and in particular $\delta_1=\ind$.

Let us now consider a smooth and compact orientable $n$-dimensional manifold $X$ without boundary and the associated short exact sequence introduced in Section \ref{Sechigh}.
The above description has the following
generalisation:
$C(X,\E)\cong C\big(X,M_2(\C)\big)\otimes C^*(\hat u)$ and
the map $\j': C\big(X,M_2(\C)\big) \to C\big(X,M_2(\C)\big)\otimes C^*(\hat u)$, $f \mapsto \j'(f)\equiv f\otimes 1$, induces an isomorphism in $K$-theory.
Furthermore, the short exact sequence \eqref{degresup} is isomorphic to the following one:
\begin{equation}\label{seqKRS}
0 \to C\big(X,M_2(\C)\big)\otimes \K\big(L^2(\R_+)\big)\stackrel{}\to C\big(X,M_2(\C)\big)\otimes C^*(\hat u) \stackrel{}\to C\big(X,M_2(\C)\big)\otimes C^*(u)\to 0\ .
\end{equation}
This exact sequence is the Toeplitz extension of the crossed
product of the algebra $C\big(X,M_2(\C)\big)$ by the trivial action of $\Z$. Note that
Pimsner and Voiculescu have considered the general case of an action
of $\Z$ on a $C^*$-algebra \cite{PV80}.
Our interest in \eqref{seqKRS} relies on the study of a more general short exact sequence performed in the Appendix of \cite{KRS}
(in that reference, the action of $\Z$ is general) and on the corresponding dual boundary maps.

2) Once this framework is settled, the next part of the proof consists in constructing a right inverse for $\ind$. The map $\jj: C\big(X,M_2(\C)\big) \to C\big(X,M_2(\C)\big)\otimes \K\big(L^2(\R_+)\big)$, $\jj(f) = f\otimes e_{00}$ induces an isomorphism in $K$-theory \cite{Roerdam}. It is hence sufficient to construct a pre-image under $\ind$ of an element of the form $[\jj(P)]_0$ where $P$ is a projection in $C\big(X,M_2(\C)\big)\otimes M_k(\C)$. Here $k$ is arbitrary and in principle higher $k$ are needed, but for simplicity of the notation we shall set $k=1$, the more general case being a simple adaptation.
Let $U\in C\big(X,M_2(\C)\big)\otimes C^*(u)$ be given by $U = 1\otimes u$, and set $U_P:=U \jj(P)+\big(1-\jj(P)\big)$. Then one has to show that $U_P$ is a unitary in $C\big(X,M_2(\C)\big)\otimes C^*(u)$ and that $\ind[U_P]_1 = [\jj(P)]_0$.
However, this calculation is well-known and in particular is performed in \cite[Prop.~A.1]{KRS}) to which we refer. Note that since the action of $\Z$ is trivial, the expression of $U_P$ introduced here is even simpler than the formula presented in that reference.

3) The last step consists in checking that the numerical equality
\begin{equation*}
\big\langle \#\eta_{X},[U_P]_1\big\rangle=
\big \langle \eta_{X},[P]_0 \big\rangle
\end{equation*}
holds. Again this direct computation has already been performed in \cite[Thm.~A.10]{KRS} to which we refer for details. Note that the constants $c_{2k}$ and $c_{2k+1}$ introduced in Section \ref{pourlesconstantes} follow from this computation.
Since the above equality has been proved for arbitrary elements of the corresponding algebras, we can then apply the result to $\bP\in C(\X,\J)$ and recall that $\bGamma\in C(X,\Q)$ is a right inverse to $-\bP$ for the map $\ind$, {\it i.e.~}$\ind[\bGamma]_1=-[\bP]_0$.
\end{proof}

\section{Proof of Lemma \ref{variationarg}}\label{appb}

Denote for brevity $\varphi=\varphi_{a,b}$.
We first observe that
\begin{equation*}
\varphi(x)^{-1} \varphi'(x)=i\Big(\frac{\Gamma'(a+ix)}{\Gamma(a+ix)}
+\frac{\Gamma'(a-ix)}{\Gamma(a-ix)}
-\frac{\Gamma'(b-ix)}{\Gamma(b-ix)}
-\frac{\Gamma'(b+ix)}{\Gamma(b+ix)}\Big).
\end{equation*}
Since the function $\Gamma$ is real on the real positive axis, let us choose a continuous
determination of the logarithm, denoted by $\log$,
such that $\log\big(\Gamma(y+ix)\big)|_{x=0}\in \R$ for any $y \in \R_+^*$.
Then, one observes that
\begin{equation*}
\varphi(x)^{-1} \varphi'(x)=\frac{\dd}{\dd x}\;\!I(x,a,b)
\end{equation*}
with
\begin{equation*}
I(x,a,b):=\log\big(\Gamma(a+ix)\big)-\log\big(\Gamma(a-ix)\big)
+\log\big(\Gamma(b-ix)\big)-\log\big(\Gamma(b+ix)\big)\ .
\end{equation*}
It follows that
\begin{equation*}
\Var[\varphi]=\frac{1}{i} \big[\lim_{x\to \infty} I(x,a,b) - \lim_{x\to -\infty}I(x,a,b)\big]
=\frac{2}{i} \lim_{x\to \infty} I(x,a,b)\ .
\end{equation*}

Now, let us denote by $\ln$ the principal determination of the logarithm,
{\it i.e.}~$\ln(z)=\ln(|z|)+i\theta(z)$, where $\theta:\C^*\to (-\pi,\pi]$ is the principal argument
of $z$.
We recall from
\cite[Eq.~6.1.37]{AS} that for $z\to \infty$ with $|\theta(z)|<\pi$:
\begin{equation*}
\Gamma(z)\cong e^{-z}\;\!e^{(z-1/2)\ln(z)}\;\!(2\pi)^{1/2} \big(1+O(z^{-1})\big)\ .
\end{equation*}
For $z = y+ix$, the term $e^{-z}\;\!e^{(z-1/2)\ln(z)}$ can be rewritten as
\begin{equation*}
e^{-y}\;\!e^{(y-1/2)\ln(\sqrt{x^2+y^2})}\;\!e^{-x\theta(y+ix)}
\exp\big\{-i\big(x-x\ln\big(\sqrt{x^2+y^2}\big)-(y-1/2)\theta(y+ix)\big)\big\}\ .
\end{equation*}
It follows that for $|x|$ large enough, one has
\begin{eqnarray*}
\log\big(\Gamma(y+ix)\big)&\cong&-y+ (y-1/2)\ln\big(\sqrt{x^2+y^2}\big)-x\theta(y+ix)+\ud \ln(2\pi) \\
&&-i\big(x-x\ln\big(\sqrt{x^2+y^2}\big)-(y-1/2)\theta(y+ix)\big)\ .
\end{eqnarray*}

By taking this asymptotic development into account, one obtains:
\begin{eqnarray*}
I(x,a,b)&\cong&
-x\big[\theta(a+ix)+\theta(a-ix)-\theta(b-ix)-\theta(b+ix)\big]\\
&&+ix\big[2\ln\big(\sqrt{a^2+x^2}\big)-2\ln\big(\sqrt{b^2+x^2}\big)\big] \\
&&+i(a-\ud)\big[\theta(a+ix)-\theta(a-ix)\big] + i(b-\ud)\big[\theta(b-ix)-\theta(b+ix)\big]\ .
\end{eqnarray*}
Clearly, for any $x$ one has
\begin{equation*}
\theta(a+ix)+\theta(a-ix)-\theta(b-ix)-\theta(b+ix)=0\ .
\end{equation*}
Furthermore, some calculations of asymptotic developments show that
\begin{equation*}
\lim_{|x|\to \infty}x\big[2\ln\big(\sqrt{a^2+x^2}\big)-2\ln\big(\sqrt{b^2+x^2}\big)\big]=0\ .
\end{equation*}
It thus follows that
\begin{eqnarray*}
&&\lim_{x\to \infty}I(x,a,b)\\
&=&i\lim_{x\to\infty}\Big\{(a-\ud)\big[\theta(a+ix)-\theta(a-ix)\big] +
(b-\ud)\big[\theta(b-ix)-\theta(b+ix)\big]\Big\}\\
&=&i(a-\ud)\big[\textstyle{\frac{\pi}{2}-\big(-\frac{\pi}{2}\big)}\big]+
i(b-\ud)\big[\textstyle{\big(-\frac{\pi}{2})-\frac{\pi}{2}}\big]\\
&=&i\pi(a-b)\ .
\end{eqnarray*}

\section{Continuity of the wave operator}\label{appc}

In this section we show that the map $X \ni U \mapsto \Omega_-(H_\alpha^U,H_0) \in \E$ is continuous under the assumptions
of Proposition~\ref{propfinal}. In view of the representations \eqref{yoyo} and \eqref{eq-stilde} for the wave operators
it is sufficient to show the continuity of the map $X \ni U \mapsto S^U\in \B$, where $\B$ is the space of bounded continuous matrix-valued functions  $S:[0,+\infty]\to M_2(\C)$ endowed with the norm
\[
\left\|
\begin{pmatrix}
s_{11}(\cdot) & s_{12}(\cdot) \\
s_{21}(\cdot) & s_{22}(\cdot)
\end{pmatrix}
\right\|=\max_{1\le j,k\le 2} \sup_{\kappa\geq 0} \big|s_{jk}(\kappa)\big|.
\]
Note that we use the notation $S^U$ for $S^{\CD}_\alpha$ with $C=C(U)$ and $D=D(U)$ defined in Remark \ref{1to1}.
Let us set
\[
L=L(U)=\dfrac{\pi}{2\sin(\pi\alpha)}\,\dfrac{1-U}{i(1+U)}=:(l_{jk})
\]
and use again the representation \eqref{eq-SL}:
\begin{equation*}
S^U(\kappa)=\Phi \;\!\frac{B^{-1}\;\! L\;\!B^{-1} +\cos(\pi\alpha)J
+i\sin(\pi\alpha)}{B^{-1}\;\! L\;\!B^{-1} +\cos(\pi\alpha)J -i\sin(\pi\alpha)}
\;\!\Phi\;\!J\ ,
\end{equation*}
with
\begin{gather*}
B\equiv B(\kappa):=\left(\begin{smallmatrix}
\frac{\Gamma(1-\alpha)}{2^\alpha}\;\!\kappa^{\alpha} & 0 \\
0 & \frac{ \Gamma(\alpha)}{2^{1-\alpha}}\;\!\kappa^{(1-\alpha)}
\end{smallmatrix}\right),
\quad
\Phi:=\left(\begin{smallmatrix}
e^{-i\pi\alpha/2} & 0 \\
0 & e^{-i\pi(1-\alpha)/2}
\end{smallmatrix}\right),\quad
J:=\left(\begin{smallmatrix}
1 & 0 \\ 0 & -1
\end{smallmatrix}\right).
\end{gather*}
Then by observing that the map $X\ni U\mapsto L(U) \in M_2(\C)$ is continuous in the usual matrix norm, it follows that the map
$$
L(X)\times [0,\infty] \ni (L,\kappa) \mapsto \frac{B(\kappa)^{-1}\;\! L\;\!B(\kappa)^{-1} +\cos(\pi\alpha)J
+i\sin(\pi\alpha)}{B(\kappa)^{-1}\;\! L\;\!B(\kappa)^{-1} +\cos(\pi\alpha)J -i\sin(\pi\alpha)} \in M_2(\C).
$$
is also continuous. This implies the required continuity of the map $X \ni U \mapsto S^U\in \B$.

\end{document}